\documentclass[journal,comsoc]{IEEEtran}
\usepackage{balance}
\usepackage{amssymb,amsthm,amsmath,epsfig,latexsym,graphicx,bm,bbm,nccmath}
\allowdisplaybreaks[4] % 0..4, higher = more willing to break
\usepackage{mathtools}

\usepackage{epstopdf}
\usepackage{booktabs}
\usepackage{tabularx}

\usepackage{tikz}
\usetikzlibrary{arrows.meta,positioning,fit,backgrounds,calc}

\usepackage{enumerate}
\usepackage{cite}
\usepackage{amsfonts}
\usepackage[ruled, vlined, linesnumbered, commentsnumbered, longend]{algorithm2e}
\SetKw{KwBreak}{break}

\usepackage{textcomp}
\usepackage{xcolor}
\def\BibTeX{{\rm B\kern-.05em{\sc i\kern-.025em b}\kern-.08em
		T\kern-.1667em\lower.7ex\hbox{E}\kern-.125emX}}
\usepackage[utf8x]{inputenc}
\usepackage{etoolbox}
\usepackage{multirow}
\usepackage[caption=false]{subfig}

\newtheorem{proposition}{Proposition}

\newtheorem{assumption}{Assumption}
\newtheorem{remark}{Remark}
\usepackage{easyReview} % for \comment, \alert, etc.

\makeatletter
\def\@IEEEfigurecaptionsepspace{\vskip 0pt\relax}
\def\@IEEEtablecaptionsepspace{\vskip 0pt\relax}
\makeatother

\begin{document}
	\bstctlcite{IEEEexample:BSTcontrol}
	\title{Privacy-Preserving Over-the-Air Federated Learning
		over Cell-Free Massive MIMO System with User-Centric Clustering}
	
	\author{Manobendu~Sarker and 
		Soumaya Cherkaoui,~\IEEEmembership{Senior Member,~IEEE}
		\thanks{M. Sarker is with Poly-Grames Research Center, Department of Electrical Engineering, Polytechnique Montr\'{e}al,  Canada. S. Cherkaoui is with the Department of Computer and Software Engineering, Polytechnique Montr{\'{e}}al, Canada. (Corresponding author's e-mail: manobendu.sarker@polymtl.ca)
	}}
	
	\maketitle
	
	
	\begin{abstract}
		This paper investigates the cross-layer design of differentially private  
		federated learning over cell-free massive MIMO networks 
		with over-the-air (OTA) aggregation. We consider a system in which 
		distributed access points (APs) collaboratively serve multiple user 
		equipments (UEs) through user-centric clustering, where local model 
		updates are clipped, perturbed with artificial Gaussian noise for privacy, 
		and transmitted simultaneously over a shared uplink for OTA averaging 
		at a central processing unit. We consider that the physical thermal noise inherent in wireless reception
		provides a natural privacy-masking resource under channel hardening,
		reducing or eliminating the need for artificial noise injection in
		the ergodic design framework. Based on this insight, we formulate a cross-layer optimization problem that jointly tunes the AP--UE association, per-UE transmit powers, and artificial noise variances to minimize the total aggregation variance subject to per-UE zero-concentrated differential privacy design constraints calibrated through an ergodic privacy accounting framework that leverages channel hardening. Analytical 
		reduction of the objective reveals a regime-dependent structure: under 
		stringent privacy requirements, the optimal design drives transmit 
		powers to their minimum and reduces AP associations to maximize the 
		thermal noise floor, effectively converting a channel impairment into 
		a free privacy resource; as privacy constraints relax, the system 
		transitions to a noise-limited regime where clustering expansion 
		minimizes interference. The resulting mixed-integer non-convex program 
		is solved via an alternating optimization framework combining 
		difference-of-convex power allocation with greedy coordinate-descent 
		clustering. Numerical evaluation demonstrates that the proposed framework 
		reduces aggregation variance by up to 89\% over conventional 
		baselines and enables federated training to achieve 97\% test 
		accuracy on MNIST while delivering consistent gains on CIFAR-10. Under 
		strict privacy budgets, all baselines remain near 
		random-chance performance, while under relaxed budgets, they 
		continue to lag significantly behind the proposed design.
	\end{abstract}
	\vspace{-6mm}
	\section{Introduction}
	Federated learning (FL) has emerged as a distributed optimization 
	framework that enables edge devices to collaboratively train machine 
	learning models without sharing their raw, private datasets 
	\cite{mcmahan2017communication}. Despite its privacy-preserving architecture, the iterative exchange of high-dimensional model updates over bandwidth-limited wireless uplinks (ULs) introduces a severe communication bottleneck. This challenge forces a fundamental tension between communication efficiency and learning fidelity. To overcome the latency and scalability issues inherent to digital orthogonal transmissions, over-the-air (OTA) computation has been widely adopted for wireless FL \cite{yang2020federated, amiri2020federated, sifaou2024ota}. By exploiting the natural superposition property of the wireless multiple-access channel, OTA-FL allows devices to transmit their analog gradients concurrently over the same time-frequency resources. The central server directly receives the aggregated model update, rendering the required bandwidth completely independent of the number of participating devices.
	
	While FL prevents the direct exposure of raw data, the transmitted local gradients remain vulnerable to inference attacks that can inadvertently reveal sensitive information. To mitigate this risk, differential privacy (DP) is commonly enforced through the clip-and-perturb Gaussian mechanism \cite{dwork2014algorithmic, abadi2016deep}, which bounds the sensitivity of local updates before injecting calibrated artificial noise. Advanced composition frameworks like zero-concentrated differential privacy (zCDP) can accurately track the cumulative privacy leakage over multiple training rounds \cite{bun2016concentrated}. However, injecting artificial noise structurally degrades gradient quality, thereby slowing convergence and lowering the final model accuracy. Existing approaches in communication systems generally treat this privacy budget as a strict, exogenous constraint \cite{liu2021privacy}. They calibrate the required artificial noise variance solely based on sensitivity and statistical privacy parameters, largely failing to couple the privacy mechanism with the physical-layer channel distortions experienced during wireless transmission.
	
	Most early OTA-FL frameworks predominantly assume isolated single-cell topologies \cite{liu2021privacy}. In dense practical deployments, however, severe inter-cell interference corrupts the analog aggregation process. While recent studies have proposed multi-cell interference management strategies \cite{wang2022interference, gesbert2010multi}, these cellular architectures fundamentally remain bottlenecked by cell-edge performance degradation and disjointed resource partitioning. To implement OTA-FL over fading wireless channels without suffering from these worst-case user equipment (UE) limitations, user-centric (UC) cell-free (CF) massive MIMO (mMIMO) provides an ideal aggregation substrate \cite{ngo2017cell, buzzi2017cell}. This architecture deploys a large number of distributed APs that cooperatively serve a set of devices through localized clusters, offering massive macro-diversity and robust channel hardening. To ensure scalability, the network operates under a distributed processing architecture where individual APs perform local spatial matched filtering \cite{bjornson2019making}. Instead of forwarding high-capacity raw baseband samples, the APs transmit only the combined scalar statistics over the fronthaul to the central processing unit. While current research optimizes power control and device clustering 
	to improve communication metrics such as the signal-to-noise  ratio (SNR), these physical-layer designs do not incorporate formal differential privacy constraints into the resource allocation formulation.
	
	Table~\ref{tab:literature_comparison} summarizes the main architectural and design features of representative prior works relative to the proposed framework. Prior privacy-aware wireless FL studies span both digital transmission and OTA aggregation. In the single-cell OTA-FL literature, existing designs optimize different subsets of the physical-layer variables, including transmit power, denoising factors, alignment coefficients, device scheduling, and aggregation rounds, under differential privacy constraints, while several explicitly exploit channel noise as a privacy resource \cite{liu2021privacy,jiang2023optimized,mao2024providing,yan2024over}. However, these frameworks still inherit the fundamental single-cell alignment bottleneck of OTA aggregation, whereby a worst-case UE can limit the global receive scaling and degrade aggregation fidelity \cite{jiang2023optimized,yan2024over}. In contrast, digital wireless-FL designs such as \cite{wei2022low} incorporate differential privacy through scheduling and channel assignment, but do not benefit from the bandwidth scalability of analog OTA aggregation. The CF OTA-FL framework in \cite{sifaou2024ota} adopts a scalable UC architecture with distributed combining, but does not incorporate DP into the resource allocation design. Multi-base-station (BS) approaches \cite{tavangaran2024privacy} and secure-aggregation-based CF learning frameworks \cite{batool2026privacy} address privacy in distributed topologies, but rely on digital transmission rather than analog OTA superposition. As shown in Table~\ref{tab:literature_comparison}, no prior work jointly optimizes UC AP--UE clustering, UL transmit power allocation, and per-UE differential privacy noise calibration within a scalable CF mMIMO OTA-FL architecture.
	
	\begin{table*}[!t]
		\centering
		\caption{Positioning of the proposed framework relative to representative prior works on private and/or OTA FL.}
		\label{tab:literature_comparison}
		\footnotesize
		\renewcommand{\arraystretch}{1.15}
		\setlength{\tabcolsep}{4.2pt}
		\begin{tabular}{@{}l c c c c c c c@{}}
			\toprule
			\textbf{Reference}
			& \textbf{Arch.}
			& \textbf{Aggr.}
			& \textbf{DP Mech.}
			& \textbf{Pwr.\ Ctrl.}
			& \textbf{UC Clust.}
			& \textbf{Therm.\ Credit}
			& \textbf{FL Valid.} \\
			\midrule
			Liu \emph{et~al.}~\cite{liu2021privacy}
			& SC & OTA & Gauss. & \checkmark & -- & \checkmark & \checkmark \\
			
			Wei \emph{et~al.}~\cite{wei2022low}
			& SC & Dig. & Gauss. & -- & -- & -- & \checkmark \\
			
			Jiang \emph{et~al.}~\cite{jiang2023optimized}
			& SC & OTA & Gauss. & \checkmark & -- & \checkmark & \checkmark \\
			
			Mao \emph{et~al.}~\cite{mao2024providing}
			& SC & OTA & Gauss.+ CJ & \checkmark & -- & \checkmark & \checkmark \\
			
			Yan \emph{et~al.}~\cite{yan2024over}
			& SC & OTA & Gauss. & -- & -- & \checkmark & \checkmark \\
			\midrule
			Sifaou \emph{et~al.}~\cite{sifaou2024ota}
			& CF & OTA & -- & -- & \checkmark$^{\star}$ & -- & \checkmark \\
			
			Tavangaran \emph{et~al.}~\cite{tavangaran2024privacy}
			& MC & Dig. & Gauss. & -- & -- & -- & \checkmark \\
			
			Batool \emph{et~al.}~\cite{batool2026privacy}
			& CF & Dig. & Gauss.+ Sec. Agg. & -- & -- & -- & --$^{\ddagger}$ \\
			\midrule
			\textbf{This work}
			& \textbf{CF} & \textbf{OTA} & \textbf{Gauss.}$^{\dagger}$
			& \checkmark & \checkmark & \checkmark & \checkmark \\
			\bottomrule
		\end{tabular}
		
		\vspace{1.2mm}
		\raggedright
		\scriptsize
		Arch.: SC\,=\,single-cell / single-BS wireless FL, MC\,=\,multi-cell / multi-BS, CF\,=\,CF mMIMO.\;
		Aggr.: OTA\,=\,OTA analog aggregation, Dig.\,=\,digital (non-OTA) transmission/aggregation.\;
		DP Mech.: Gauss.\,=\,Gaussian DP mechanism or Gaussian-noise-based DP accounting, CJ\,=\,cooperative jammer, Sec. Agg.\,=\,secure aggregation / cryptographic masking, --\,=\,not addressed.\;
		Pwr.\ Ctrl.\,=\,privacy-aware transmit-power optimization.\;
		UC Clust.\,=\,UC AP--UE clustering/association explicitly considered in the system design.\;
		Therm.\ Credit\,=\,physical channel/thermal noise explicitly credited toward the privacy budget or privacy guarantee.\;
		FL Valid.\,=\,end-to-end (E2E) federated learning task validation reported.\;
		$^{\dagger}$Ergodic zCDP design proxy calibrated under channel hardening (see Section~\ref{subsec:zcdp}).\;
		$^{\star}$UC clustering is part of the CF architecture and AP--UE association is explicitly studied, but not jointly optimized with DP.\;
		$^{\ddagger}$Federated meta-learning is validated for wireless resource allocation, but not through a standard E2E FL prediction task such as MNIST/CIFAR classification.
		\vspace{-4mm}
	\end{table*}
	\vspace{-2mm}
	\subsection{Motivation and Contributions}
	
	The distributed processing architecture inherently blinds the central server to individual UL signals, offering a significant yet overlooked physical-layer privacy advantage at the design level. Specifically, the residual thermal noise generated at the AP antennas is functionally identical to the artificial Gaussian perturbation required for DP \cite{liu2021privacy}. When devices operate under strict transmit power limits, severe channel fading traditionally forces the network to either drop vulnerable UEs or drastically shrink the global alignment coefficient, thereby impacting the received SNR. Although recent studies have attempted to resolve this issue through sequential device scheduling and optimized power control \cite{cao2022optimized, wei2022low}, they continue to treat the privacy requirement as a rigid, exogenous constraint. Existing frameworks typically maximize transmit power to suppress thermal noise for better communication reliability, and then subsequently inject artificial noise of comparable magnitude to satisfy DP metrics \cite{wei2022low, jiang2023optimized}. This approach imposes a double penalty on the network, wasting 
	transmit energy while needlessly degrading the learning process. The current literature lacks a unified framework that recognizes distributed thermal noise as a free privacy resource.
	
	To address this fundamental limitation, we propose a cross-layer, privacy-preserving OTA-FL framework tailored for a UC CF mMIMO architecture. Our approach differs from existing methodologies along the three axes highlighted in Table~\ref{tab:literature_comparison}: 
	(i) unlike single-cell DP-aware OTA-FL frameworks \cite{mao2024providing,yan2024over}, the CF topology alleviates the worst-case UE alignment bottleneck through macro-diversity; 
	(ii) unlike digital multi-BS designs \cite{tavangaran2024privacy} and secure-aggregation-based CF learning frameworks \cite{batool2026privacy}, the proposed system exploits analog OTA superposition for bandwidth-scalable aggregation; and 
	(iii) unlike the CF OTA-FL design in \cite{sifaou2024ota}, our framework couples physical-layer resource allocation with differential privacy constraints through an ergodic privacy calibration supported by channel hardening.
	
	As established by the comparative analysis in 
	Table~\ref{tab:literature_comparison}, the joint treatment of 
	UC clustering, UL power control, and differential privacy 
	noise calibration within a scalable CF mMIMO OTA-FL architecture 
	has not been addressed in the existing literature. This work 
	fills this gap. The main contributions of this paper are summarized as follows:
	\begin{enumerate}
		\item We derive a 
		convergence bound for non-convex FL over imperfect OTA channels that 
		establishes a tractable, variance-based design objective. The bound 
		maps the CF OTA distortion structure into an aggregation 
		variance surrogate that explicitly incorporates multi-UE 
		interference and channel state information (CSI) error proxies under the distributed combining 
		architecture. The novelty lies not in the descent analysis itself, 
		which follows standard smoothness arguments, but in the systematic 
		identification of the physically controllable variance components 
		that constitute the cross-layer optimization target.
		\item We develop an ergodic zCDP design-based proxy that maps system-wide thermal noise directly into the privacy budget. By utilizing Chebyshev concentration bounds, we characterize 
		the reliability with which channel hardening can control the 
		instantaneous sensitivity, thereby motivating an ergodic 
		privacy calibration and its robust counterpart.
		\item We 
		formulate a joint optimization of UC clustering, transmit powers, 
		and artificial noise variances. By analytically eliminating the 
		artificial noise through the tight privacy constraint, we unveil 
		a two-regime structural insight wherein the privacy-limited regime 
		drives transmit powers down to substitute free physical thermal 
		noise for costly artificial noise. We exploit this structure through 
		an alternating optimization framework combining a difference-of-convex (DC) power allocation solved via the convex-concave procedure (CCP) with a greedy coordinate-descent clustering update, yielding a monotonically convergent alternating optimization (AO) procedure whose terminal point is blockwise locally optimal with respect to the CCP power update and accepted greedy clustering modifications.
		\item  We validate the framework 
		across multiple datasets and privacy regimes. The results confirm 
		that the proposed cross-layer design enables a qualitative transition 
		from random-chance to convergent training under strict privacy budgets where conventional schemes fail.
	\end{enumerate}
	
	\vspace{-2mm}
	\section{System Model}
	\label{sec:system}
	
	We consider a UC CF mMIMO network supporting synchronous FL over a shared wireless UL \cite{sifaou2024ota}. The network consists of $N$ single-antenna UEs indexed by $\mathcal{N}=\{1,\dots,N\}$ and $L$ distributed APs indexed by $\mathcal{L}=\{1,\dots,L\}$, each equipped with $M$ antennas. All APs are connected to a central processing unit (CPU) through high-capacity fronthaul links. To isolate the cross-layer interactions among clustering, power control, OTA distortion, and DP, we assume full participation of all UEs in each communication round.
	\vspace{-2mm}
	\subsection{UC Clustering and Channel Model}
	\label{subsec:clustering_channel}
	
	To ensure scalability and respect fronthaul capacity limitations, each UE is served by a dynamically selected subset of APs \cite{buzzi2017cell}. We define binary association variables $a_{m,k}\in\{0,1\}$, where $a_{m,k}=1$ indicates that AP $m$ participates in processing the UL signal of UE $k$. Due to hardware constraints, we consider that each AP serves at most $U_{\max}$ UEs per round, yielding the load limit $\sum_{k=1}^{N} a_{m,k} \leq U_{\max}$ for all $m\in\mathcal{L}$. Similarly, each UE is served by at most $L_{\max}$ APs, enforced by $\sum_{m=1}^{L} a_{m,k} \leq L_{\max}$ for all $k\in\mathcal{N}$.
	
	UL channels follow a block-fading model \cite{ngo2017cell, buzzi2017cell}. The channel from UE $k$ to AP $m$ in round $t$ is modeled as $\mathbf{h}_{m,k}^{(t)} \sim \mathcal{CN}(\mathbf{0}, \mathbf{R}_{m,k})$, where the spatial correlation matrix $\mathbf{R}_{m,k}\in\mathbb{C}^{M\times M}$ captures large-scale fading, shadowing, and antenna correlation. Each AP estimates the local channels using UL pilot sequences via minimum mean-square error (MMSE) estimation. The true channel decomposes as $\mathbf{h}_{m,k}^{(t)} = \hat{\mathbf{h}}_{m,k}^{(t)} + \tilde{\mathbf{h}}_{m,k}^{(t)}$, where $\hat{\mathbf{h}}_{m,k}^{(t)}$ is the channel estimate and $\tilde{\mathbf{h}}_{m,k}^{(t)}\sim \mathcal{CN}(\mathbf{0}, \mathbf{C}_{m,k})$ is the independent estimation error. The covariance of the estimate is $\boldsymbol{\Phi}_{m,k} = \mathbf{R}_{m,k}-\mathbf{C}_{m,k}$. This structure inherently captures pilot contamination effects when UEs share non-orthogonal training sequences.
	
	To quantify link quality while rigorously accounting for channel estimation accuracy, the effective usable gain of UE $k$ under the association vector $\mathbf{a}_k=[a_{1,k},\dots,a_{L,k}]^T$ is defined as $G_k(\mathbf{a}_k) \triangleq \sum_{m=1}^{L} a_{m,k}\,\mathrm{tr}\!\left(\boldsymbol{\Phi}_{m,k}\right)$. Because clustering and power allocation target long-term performance, the physical-layer variables $a_{m,k}$ and $p_k$ are optimized based on large-scale fading statistics, where $p_k$ is the data transmit power of UE $k$. We assume the CPU computes these optimal variables and broadcasts the corresponding normalization coefficients to the APs, keeping the large-scale configuration static across the training horizon.
	\vspace{-2mm}
	\subsection{Local Computation and DP}
	\label{subsec:local_dp}
	
	At round $t$, UE $k$ executes local training to compute a model update $\mathbf{g}_k^{(t)}\in\mathbb{R}^{d}$. To support sample-level DP, the UE applies a clip-and-perturb Gaussian mechanism prior to transmission \cite{dwork2014algorithmic}. The update is clipped to an $\ell_2$-norm bound $C>0$ to limit sensitivity: $\bar{\mathbf{g}}_k^{(t)} = \mathbf{g}_k^{(t)}\cdot \min\!\left(1,\frac{C}{\|\mathbf{g}_k^{(t)}\|_2}\right)$. Subsequently, the UE injects independent Gaussian noise to form the privatized update $\tilde{\mathbf{g}}_k^{(t)} = \bar{\mathbf{g}}_k^{(t)} + \mathbf{n}_k^{\mathrm{DP}}(t)$, where $\mathbf{n}_k^{\mathrm{DP}}(t)\sim\mathcal{N}(\mathbf{0},\sigma_k^2\mathbf{I}_d)$, and $\sigma_k^2$ is the allocated artificial noise variance. We collect these into the vector 
	$\boldsymbol{\sigma}^2 = [\sigma_1^2, \ldots, \sigma_N^2]^T$.
	\vspace{-2mm}
	\subsection{Scalable OTA Aggregation}
	\label{subsec:ota}
	
	To transmit the real-valued updates over the complex baseband channel, pairs of coordinates are mapped to the in-phase and quadrature components of complex symbols. For analytical clarity, we express the transmission and variance models on a per-real-coordinate basis, dropping the coordinate index. Letting $\sigma_n^2$ denote the thermal noise variance per complex dimension, the received baseband signal at AP $m$ is the $M \times 1$ complex vector
	\begin{equation}
		\mathbf{y}_m^{(t)} = \sum_{k=1}^{N} \mathbf{h}_{m,k}^{(t)} x_k^{(t)} + \mathbf{z}_m^{(t)},
		\label{eq:rx_ap}
	\end{equation}
	where $\mathbf{z}_m^{(t)}\sim\mathcal{CN}(\mathbf{0},\sigma_n^2\mathbf{I}_M)$. UE $k$ transmits the scalar symbol $x_k^{(t)} = \sqrt{p_k}\,\tilde{g}_k^{(t)}$, subject to the power constraint $0 \le p_k \le P_{\max}$. For downstream variance modeling on a per-real-dimension basis, we define the effective thermal noise variance as $\sigma_w^2 = \sigma_n^2/2$. 
	
	This architecture aligns with \cite{bjornson2019making}, in which APs perform local combining and forward only scalar statistics to the CPU rather than raw baseband samples, thereby reducing fronthaul requirements and supporting scalable deployment. Specifically, AP $m$ constructs a local combiner $\mathbf{v}_m^{(t)} = \sum_{k=1}^{N} a_{m,k}\,c_k\,\hat{\mathbf{h}}_{m,k}^{(t)}$, where $c_k = (N\sqrt{p_k}\,G_k(\mathbf{a}_k))^{-1}$ is a deterministic spatial equalization coefficient designed to invert the expected large-scale gain. AP $m$ calculates the scalar $s_m^{(t)} = \mathbf{v}_m^{(t)H} \mathbf{y}_m^{(t)}$ and forwards only this single value to the CPU over the fronthaul.
	
	The CPU aggregates the scalars to form the final observation $\hat{g}^{(t)} = \sum_{m=1}^{L} s_m^{(t)}$. Under channel hardening, the effective aggregate decomposes into the desired average and residual physical-layer distortions \cite{sifaou2024ota}, given by
	\begin{equation}
		\hat{g}^{(t)} \approx \frac{1}{N}\sum_{k=1}^{N}\tilde{g}_k^{(t)} + e_{\mathrm{int}}^{(t)} + e_{\mathrm{csi}}^{(t)} + e_{\mathrm{th}}^{(t)},
		\label{eq:ota_decomp}
	\end{equation}
	where $e_{\mathrm{int}}^{(t)}$, $e_{\mathrm{csi}}^{(t)}$, $e_{\mathrm{th}}^{(t)}$ are the errors due to multi-UE interference, CSI, thermal noise, respectively. By evaluating the expectations of the quadratic forms induced by the sum-matched filtering, we map the post-combining distortion to the scale of a single UE. Because the wireless channel is a broadcast medium, UE $j$'s signal physically reaches AP $m$ regardless of the clustering assignment $a_{m,j}$. The multi-UE interference proxy operating at the CPU aggregation scale evaluates to
	\begin{equation}
		e_{\mathrm{int},k}(\mathbf{A},\mathbf{p}) \le \frac{1}{N^2 p_k\,G_k(\mathbf{a}_k)^2} \sum_{j\neq k} p_j \sum_{m=1}^{L} a_{m,k}\,\mathrm{tr}\!\big(\boldsymbol{\Phi}_{m,k}\mathbf{R}_{m,j}\big).
		\label{eq:e_int_proxy}
	\end{equation}
	
	The CSI error variance captures the distortion from imperfect knowledge of UE~$k$'s own channel and depends solely on the clustering decisions. Its dominant self-term evaluates to
	\begin{equation}
		e_{\mathrm{csi},k}(\mathbf{A}) \le \frac{1}{N^2\,G_k(\mathbf{a}_k)^2} \sum_{m=1}^{L} a_{m,k}\,\mathrm{tr}\!\big(\boldsymbol{\Phi}_{m,k}\mathbf{C}_{m,k}\big).
		\label{eq:e_csi_proxy}
	\end{equation}
	Cross-UE CSI leakage terms proportional to $\mathrm{tr}(\boldsymbol{\Phi}_{m,j}\mathbf{C}_{m,k})$ for $j\neq k$ are already absorbed by the multi-UE interference proxy \eqref{eq:e_int_proxy}, which uses the total channel covariance $\mathbf{R}_{m,j}=\boldsymbol{\Phi}_{m,j}+\mathbf{C}_{m,j}$; accounting for them separately would double-count their effect. A formal derivation showing the power-independence of the self-term is provided in Appendix~\ref{app:csi_derivation}.
	
	Crucially, the aggregate thermal noise affecting the CPU observation is a single, shared random variable. Since the thermal noise vectors are independent across APs, the conditional variance is strictly $\sigma_w^2 \sum_{m=1}^{L} \|\mathbf{v}_m\|^2$. Taking the expectation over the large-scale fading timescale, and neglecting positive cross-UE estimate correlations to form a tractable lower-bound proxy under pilot contamination, the expected system-wide thermal noise variance is defined as
	\begin{equation}
		\Sigma(\mathbf{A},\mathbf{p}) \triangleq \frac{\sigma_w^2}{N^2} \sum_{j=1}^N \frac{1}{p_j G_j(\mathbf{a}_j)}.
		\label{eq:e_th_proxy}
	\end{equation}
	
	The proxy in~\eqref{eq:e_th_proxy} is obtained by evaluating the 
	expected combining norm $\mathbb{E}[\|\mathbf{v}_m\|^2]$ and 
	retaining only the diagonal self-terms $a_{m,j}^2 c_j^2 
	\,\mathrm{tr}(\boldsymbol{\Phi}_{m,j})$. The omitted cross-UE 
	terms $a_{m,j}\,a_{m,j'}\,c_j\,c_{j'}\,\mathrm{tr}
	(\boldsymbol{\Phi}_{m,j}\boldsymbol{\Phi}_{m,j'})$ for $j \neq j'$ 
	are nonnegative under the positive semidefinite covariance 
	structure, so their exclusion yields a lower bound on the true 
	expected thermal noise variance. This introduces a dual-use 
	tension in the cross-layer design. As a lower bound, 
	$\Sigma(\mathbf{A},\mathbf{p})$ provides a safe underestimate 
	when used in the privacy 
	constraint~\eqref{eq:privacy_constraint_local}, ensuring that the 
	credited thermal shielding does not exceed what the channel 
	provides in expectation. However, the same bound is structurally 
	optimistic when evaluating convergence performance 
	through~\eqref{eq:convergence_bound}, since the true aggregation 
	noise may exceed the proxy.
	
	The analysis in this paper isolates the coupling between privacy 
	accounting, power control, and UC clustering under 
	several modeling assumptions. Full UE participation is assumed in 
	every communication round, and the large-scale optimization 
	variables $(\mathbf{A}, \mathbf{p}, \boldsymbol{\sigma}^2)$ remain 
	static across the training horizon. The fronthaul links between 
	APs and CPU are treated as reliable and delay-free. The 
	per-coordinate scalar transmission model in 
	Section~\ref{subsec:ota} abstracts the mapping of real-valued 
	updates to complex baseband symbols through in-phase/quadrature (I/Q) packing. While 
	the channel model in Section~\ref{subsec:clustering_channel} 
	accommodates pilot contamination through the general covariance 
	structure $\mathbf{R}_{m,k} = \boldsymbol{\Phi}_{m,k} + 
	\mathbf{C}_{m,k}$, the numerical evaluation employs orthogonal 
	pilots with $\tau_p = N$, under which the cross-UE estimate 
	correlations are suppressed and the proxy accuracy improves. Extensions 
	to partial UE participation, dynamic large-scale adaptation, 
	pilot reuse, and fronthaul capacity constraints constitute 
	directions for future work.
	
	\section{Privacy and Learning Preliminaries}
	\label{sec:prelim}
	
	\subsection{Privacy Accounting and Ergodic Design Proxy}
	\label{subsec:zcdp}
	
	We assume the CPU is honest-but-curious and observes the sequence $\{\hat{g}^{(t)}\}_{t=1}^{T}$. An explicit security advantage of the distributed scalar-forwarding architecture in  \cite{bjornson2019making} is that raw baseband samples are never forwarded over the fronthaul, physically blinding the CPU to individual UL streams.
	
	If a single data sample in UE $k$'s dataset is altered, the clipped update $\bar{\mathbf{g}}_k^{(t)}$ changes by a maximum of $C$ in $\ell_2$-norm. Under channel hardening, the instantaneous effective coefficient $\sqrt{p_k}\,\alpha_k^{\mathrm{self}} = S_k/(N G_k)$, where $S_k = \sum_{m=1}^L a_{m,k}\|\hat{\mathbf{h}}_{m,k}\|^2$, concentrates tightly around $1/N$. The resulting per-coordinate sensitivity satisfies $\Delta_{\mathrm{inst}} = C\,S_k/(N G_k)$, with $\mathbb{E}[\Delta_{\mathrm{inst}}] = C/N$. A quantitative bound on the concentration is given next.
	
	\begin{remark}
		\label{rem:sensitivity}
		Define the channel-hardening quality metric
		\begin{equation}
			\beta_k \triangleq \frac{G_k(\mathbf{a}_k)^2}{\sum_{m=1}^{L} a_{m,k}\,\mathrm{tr}\!\bigl(\boldsymbol{\Phi}_{m,k}^2\bigr)}.
			\label{eq:beta_k}
		\end{equation}
		By Chebyshev's inequality applied to $S_k = \sum_m a_{m,k}\|\hat{\mathbf{h}}_{m,k}\|^2$, the instantaneous sensitivity satisfies
		\begin{equation}
			\Pr\!\left[\Delta_{\mathrm{inst}} \ge (1+\xi)\frac{C}{N}\right] \le \frac{1}{\xi^2\,\beta_k}, \quad \forall\,\xi > 0,
			\label{eq:sensitivity_chebyshev}
		\end{equation}
		where $\xi$ is the deviation factor. Under isotropic spatial correlation $\boldsymbol{\Phi}_{m,k}=\phi_{m,k}\mathbf{I}_M$, the metric evaluates to $\beta_k = M\,(\textstyle\sum_m a_{m,k}\phi_{m,k})^2 \big/ \sum_m a_{m,k}\phi_{m,k}^2$, which satisfies $\beta_k \ge M$ with equality when $|\mathcal{L}_k|=1$. The failure probability therefore scales as $\mathcal{O}(1/M)$, confirming that massive arrays provide the concentration required for the ergodic design proxy. Over $T$ rounds, a union bound yields a cumulative failure probability of at most $T/(\xi^2\beta_k)$.
	\end{remark}
	
	To rigorously track cumulative privacy loss, we apply zCDP \cite{bun2016concentrated}. Let $\rho_k^{\text{tot}}$ denote the total zCDP budget corresponding to a target $(\varepsilon_k, \delta)$-DP guarantee, converted via $\varepsilon_k = \rho_k^{\text{tot}} + 2\sqrt{\rho_k^{\text{tot}} \ln(1/\delta)}$. Here, $\varepsilon_k$ denotes the target privacy-loss parameter for UE $k$, and $\delta$ denotes the allowable probability of privacy failure in the standard $(\varepsilon_k,\delta)$-DP sense. Assuming identical privacy allocations per round, the mechanism must satisfy a per-round budget of $\rho_k^{\text{per}} \le \rho_k^{\text{tot}}/T$.
	
	Following standard DP conventions, the privacy accounting relies exclusively on data-independent noise; thus, data-dependent interference is strictly excluded from the privacy budget. Furthermore, to ensure an autonomous privacy guarantee that does not rely on other UEs properly participating or injecting artificial noise, we conservatively account only for UE $k$'s local artificial noise and the shared physical thermal noise. 
	
	While strict realization-wise DP accounting would require a deterministic lower bound on the instantaneous thermal noise variance, dynamic adaptation of physical resources at the small-scale fading timescale is infeasible. Therefore, for long-term resource allocation, we formulate an \emph{ergodic DP design proxy}, calibrating the privacy constraint using the expected thermal noise variance $\Sigma(\mathbf{A},\mathbf{p})$. 
	
	The total expected Gaussian noise shielding the scaled update of UE $k$ in expectation is $\sigma_{\mathrm{eff},k}^2 = \frac{\sigma_k^2}{N^2} + \Sigma(\mathbf{A},\mathbf{p})$. For this combined mechanism, the per-round zCDP parameter is $\rho_k^{\text{per}} = (C/N)^2 / (2\sigma_{\mathrm{eff},k}^2)$. Enforcing $\rho_k^{\text{per}} \le \rho_k^{\text{tot}}/T$ yields
	\begin{equation}
		\frac{\sigma_k^2}{N^2} + \Sigma(\mathbf{A},\mathbf{p}) \ge \frac{T(C/N)^2}{2\rho_k^{\text{tot}}}.
		\label{eq:privacy_bound_cpu}
	\end{equation}
	Consequently, multiplying by $N^2$ aligns the constraint to the local UE scale, resulting in
	\begin{equation}
		\sigma_k^2 + N^2 \Sigma(\mathbf{A},\mathbf{p}) \ge \Theta_k \triangleq \frac{T C^2}{2\rho_k^{\text{tot}}}.
		\label{eq:privacy_constraint_local}
	\end{equation}
	
	\begin{remark}\label{rem:privacy_scope}
		The privacy constraint in~\eqref{eq:privacy_constraint_local} is 
		calibrated using the expected thermal noise variance 
		$\Sigma(\mathbf{A},\mathbf{p})$, which serves as an ergodic 
		design proxy. The resulting privacy accounting provides an 
		{ergodic zCDP design guarantee} supported by channel 
		hardening, which should not be interpreted as a strict 
		realization-wise $(\varepsilon,\delta)$-DP theorem.
		
		The distinction arises because the instantaneous thermal noise 
		at the CPU, $\sigma_w^2\sum_{m=1}^{L}\|\mathbf{v}_m\|^2$, 
		fluctuates around its expectation~$\Sigma(\mathbf{A},\mathbf{p})$. 
		By Remark~\ref{rem:sensitivity}, the channel-hardening quality 
		metric~$\beta_k$ governs this concentration: the per-round 
		Chebyshev failure probability scales as $1/(\xi^2\beta_k)$, 
		which vanishes as $\mathcal{O}(1/M)$ under isotropic 
		correlation but may remain non-negligible for small~$M$ or 
		sparse AP associations. To provide a tunable robustness 
		mechanism against such finite-$M$ fluctuations, a conservative 
		margin is introduced through
		\begin{equation}\label{eq:robust_proxy}
			\Sigma_{\mathrm{rob}}(\mathbf{A},\mathbf{p}) 
			\triangleq \kappa\,\Sigma(\mathbf{A},\mathbf{p}), 
			\qquad 0 < \kappa \leq 1.
		\end{equation}
		Replacing $\Sigma$ with $\Sigma_{\mathrm{rob}}$ in 
		\eqref{eq:privacy_constraint_local} yields the robust counterpart
		\begin{equation}\label{eq:robust_constraint}
			\sigma_k^2 + N^2\Sigma_{\mathrm{rob}}(\mathbf{A},\mathbf{p}) 
			\geq \Theta_k, \qquad \forall k \in \mathcal{N},
		\end{equation}
		which reduces the credited thermal-noise by a factor~$\kappa$, 
		requiring additional artificial noise of 
		$(1-\kappa)N^2\Sigma(\mathbf{A},\mathbf{p})$ but tightening 
		the effective privacy calibration. The parameter~$\kappa$ thus 
		controls the tradeoff between thermal-noise credit and design 
		conservatism: $\kappa=1$ represents the ergodic limit, while 
		smaller values systematically strengthen the privacy assurance 
		against instantaneous fading excursions. This substitution 
		preserves the optimization structure developed in 
		Section~\ref{sec:solution}, and Section~\ref{sec:hardening} quantifies the conservatism of this 
		concentration certificate under the simulated deployment 
		and evaluates the resulting robustification through~$\kappa$.
	\end{remark}
	
	Beyond the interpretation of the thermal-noise shielding term, the
	privacy requirement in~\eqref{eq:privacy_constraint_local} also depends explicitly on the
	clipping threshold \(C\), which introduces the following standard
	trade-off. Smaller values of \(C\) reduce the required DP noise
	\(\Theta_k \propto C^2\), but they also increase the systematic bias
	induced during learning. The treatment of $C$ as a fixed hyperparameter and its
	implications for the cross-layer objective are discussed further
	in Remark~\ref{rem:clipping_bias}. 
	\vspace{-3mm}
	\subsection{Learning Model and Utility Surrogate}
	\label{subsec:convergence}
	
	We consider the FL objective $F(\mathbf{w}) = \sum_{k=1}^{N} \pi_k F_k(\mathbf{w})$ optimized via FedAvg \cite{mcmahan2017communication}, where $\pi_k \geq 0$ denotes the aggregation weight of UE~$k$. To formulate the optimization objective, we establish standard conditions on the learning environment.
	\begin{assumption}
		\label{ass:smoothness}
		The global loss $F(\mathbf{w})$ is $L_s$-smooth and bounded below by $F(\mathbf{w}^*)$.
	\end{assumption}
	\begin{assumption}
		\label{ass:variance}
		Local stochastic gradients are unbiased with bounded variance $\sigma_g^2$, and the gradient heterogeneity across UEs is bounded by $\zeta^2$.
	\end{assumption}
	\begin{assumption}
		\label{ass:clipping}
		The clipping threshold $C$ is sufficiently large that the systematic bias introduced by gradient clipping is negligible.
	\end{assumption}
	
	Under these assumptions, we characterize the impact of the communication framework on the learning algorithm.
	
	\begin{proposition}
		\label{prop:convergence}
		Let $\eta_g \le 1/L_s$ and assume a single local stochastic gradient descent (SGD) step per round ($\tau=1$). Under Assumptions~\ref{ass:smoothness}--\ref{ass:clipping} and the additive OTA aggregation noise model in~\eqref{eq:ota_decomp}, the average squared gradient norm satisfies
		\begin{equation}
			\frac{1}{T}\sum_{t=0}^{T-1} \mathbb{E}\!\left[\left\|\nabla F(\mathbf{w}^{(t)})\right\|^2\right] \le \frac{2\Delta_0}{T\eta_g} + L_s\eta_g\!\left(\frac{\sigma_g^2}{N} + \mathrm{Var}_{\mathrm{agg}}\right),
			\label{eq:convergence_bound}
		\end{equation}
		where $\Delta_0 = F(\mathbf{w}^{(0)})-F(\mathbf{w}^*)$ is the initial optimality gap and $\mathrm{Var}_{\mathrm{agg}} = \frac{1}{N^2}\sum_{k=1}^{N}\sigma_k^2 + \Sigma(\mathbf{A},\mathbf{p}) + \sum_{k=1}^{N}\bigl(e_{\mathrm{int},k}(\mathbf{A},\mathbf{p}) + e_{\mathrm{csi},k}(\mathbf{A})\bigr)$.
	\end{proposition}
	
	\begin{proof}
		See Appendix~\ref{app:convergence_proof}.
	\end{proof}
	
	\begin{remark}
		\label{rem:clipping_bias}
		Assumption~\ref{ass:clipping} treats the systematic bias introduced 
		by gradient clipping as negligible. In practice, clipping at 
		threshold $C$ may introduce an additional bias-dependent term in 
		the standard smoothness-based convergence analysis. Since $C$ is 
		treated here as a fixed learning hyperparameter, this bias is 
		independent of the physical-layer variables 
		$(\mathbf{A}, \mathbf{p}, \boldsymbol{\sigma}^2)$. Consequently, 
		minimizing $\mathrm{Var}_{\mathrm{agg}}$ remains the correct 
		cross-layer design objective regardless of the clipping regime.
		The present formulation isolates the communication-privacy coupling 
		by fixing $C$, consistent with standard practice in differentially 
		private FL.
	\end{remark}
	
	The noise floor in~\eqref{eq:convergence_bound} is proportional to $L_s\eta_g\,\mathrm{Var}_{\mathrm{agg}}$. Since $L_s$, $\eta_g$, $N$, and $\sigma_g^2$ are predetermined, minimizing the gradient-norm floor reduces to minimizing $\mathrm{Var}_{\mathrm{agg}}$. Scaling by $N^2$ for unit consistency with the locally injected noise yields the cross-layer surrogate, defined as
	\begin{align}
		\mathcal{U}(\mathbf{A},\mathbf{p},\boldsymbol{\sigma}^2) = & \sum_{k=1}^{N}\sigma_k^2 + N^2\Sigma(\mathbf{A},\mathbf{p}) \nonumber\\ 
		& + N^2\sum_{k=1}^{N}\bigl(e_{\mathrm{int},k}(\mathbf{A},\mathbf{p}) + e_{\mathrm{csi},k}(\mathbf{A})\bigr).
		\label{eq:utility_surrogate}
	\end{align}
	Because the thermal noise $\Sigma$ is a single shared physical variable, it contributes exactly once to the scaled surrogate.
	\vspace{-2mm}
	\section{Problem Formulation}
	\label{sec:problem}
	
	Our objective is to minimize the cross-layer surrogate 
	$\mathcal{U}(\mathbf{A}, \mathbf{p}, \boldsymbol{\sigma}^2)$ by jointly 
	optimizing the UC clustering matrix $\mathbf{A}$, the transmit 
	power allocation $\mathbf{p}$, and the artificial noise variances 
	$\boldsymbol{\sigma}^2$. This objective follows directly from 
	Proposition~\ref{prop:convergence}: since the steady-state gradient-norm 
	floor is governed by the additive term $L_s\eta_g\,\mathrm{Var}_{\mathrm{agg}}$, minimizing 
	$\mathcal{U}$ maximizes gradient recovery fidelity at the CPU and thereby 
	accelerates FL convergence. The joint optimization problem subject to 
	physical-layer and ergodic privacy design constraints is formulated as follows:
	\begin{subequations}
		\begin{align}
			\min_{\mathbf{A},\mathbf{p},\boldsymbol{\sigma}^2} \quad 
			& \mathcal{U}(\mathbf{A},\mathbf{p},\boldsymbol{\sigma}^2)
			\label{eq:prob_obj}\\
			\text{s.t.} \quad 
			& \sigma_k^2 + N^2 \Sigma(\mathbf{A},\mathbf{p}) \ge \Theta_k, \quad \forall k\in\mathcal{N}, \label{eq:dp_constraint}\\
			& \sigma_k^2 \ge 0, \quad \forall k\in\mathcal{N}, \label{eq:noise_pos_constraint}\\
			& 0\le p_k \le P_{\max}, \quad \forall k\in\mathcal{N}, \label{eq:power_constraint}\\
			& \sum_{m=1}^{L} a_{m,k} \le L_{\max}, \quad \forall k\in\mathcal{N}, \label{eq:cluster_constraint}\\
			& \sum_{k=1}^{N} a_{m,k} \le U_{\max}, \quad \forall m\in\mathcal{L}, \label{eq:load_constraint}\\
			& a_{m,k}\in\{0,1\}, \quad \forall m\in\mathcal{L},\ \forall k\in\mathcal{N}. \label{eq:binary_constraint}
		\end{align}
	\end{subequations}
	
	In problem \eqref{eq:prob_obj}, constraint \eqref{eq:dp_constraint} stems from \eqref{eq:robust_constraint} in Remark \ref{rem:privacy_scope}, enforcing the privacy requirement for each UE, while constraint \eqref{eq:noise_pos_constraint} imposes the nonnegativity of the artificial noise variance. Constraint \eqref{eq:power_constraint} bounds the transmit power of each UE by \(P_{\max}\). Finally, as described in Section \ref{subsec:clustering_channel}, constraints \eqref{eq:cluster_constraint} and \eqref{eq:load_constraint} limit the number of APs associated with each UE to \(L_{\max}\) and the number of UEs served by each AP to \(U_{\max}\), respectively, and constraint \eqref{eq:binary_constraint} specifies the binary nature of the AP--UE association decision variables.
	
	Since the objective is strictly increasing in $\sigma_k^2$, the optimal noise variance at the minimum resides precisely at the boundary defined by \eqref{eq:dp_constraint} and \eqref{eq:noise_pos_constraint}. The artificial noise is analytically eliminated via the relation $\sigma_k^{2\star} = \max\!\left(0, \Theta_k - N^2 \Sigma(\mathbf{A},\mathbf{p})\right)$, yielding the reduced problem operating exclusively on the physical variables $(\mathbf{A}, \mathbf{p})$:
	\begin{subequations}
		\begin{align}
			\min_{\mathbf{A},\mathbf{p}} \quad & \sum_{k=1}^{N} \max\!\left(0, \Theta_k - N^2 \Sigma(\mathbf{A},\mathbf{p})\right) + N^2 \Sigma(\mathbf{A},\mathbf{p}) \nonumber \\
			& + N^2\sum_{k=1}^{N} \big(e_{\mathrm{int},k}(\mathbf{A},\mathbf{p}) + e_{\mathrm{csi},k}(\mathbf{A})\big)
			\label{eq:reduced_problem}\\
			\text{s.t.} \quad & \quad \eqref{eq:power_constraint}\text{--}\eqref{eq:binary_constraint}.
		\end{align}
	\end{subequations}
	
	The reduced objective in \eqref{eq:reduced_problem} exposes a key structural insight regarding network operation. In a \emph{fully privacy-limited regime} where the thermal noise floor is lower than the strictest privacy requirement ($N^2 \Sigma < \min_k \Theta_k$), the max operator resolves to $\Theta_k - N^2 \Sigma$ for all UEs. The DP and thermal noise components of the objective evaluate to $\sum_k \Theta_k - (N-1)N^2\Sigma$. Because the coefficient of the thermal noise is strictly negative, the optimizer is actively driven to {increase} the physical thermal noise. Since scaling all powers down by a uniform factor $\alpha < 1$ increases $\Sigma$ proportionately as $1/\alpha$ while leaving the interference $e_{\mathrm{int},k}$ structurally unchanged, the optimal strategy actively drives transmit powers down until the regime boundary $N^2 \Sigma = \min_k \Theta_k$ is reached. This behavior yields a favorable operating point in which the network reduces transmit energy while replacing part of the required artificial privacy noise with inherent thermal noise.
	
	\begin{proposition}
		\label{prop:scaling_regime}
		For a fixed feasible clustering matrix \(\mathbf{A}\) and a uniformly scaled power vector \(\mathbf{p}' = \alpha \mathbf{p}\) with \(0 < \alpha < 1\), the thermal-noise proxy and the interference proxy bounds satisfy
		\begin{align}
			\Sigma(\mathbf{A},\mathbf{p}')
			&=
			\frac{1}{\alpha}\Sigma(\mathbf{A},\mathbf{p}),
			\label{eq:scaling_sigma_rev}\\
			e_{\mathrm{int},k}(\mathbf{A},\mathbf{p}')
			&=
			e_{\mathrm{int},k}(\mathbf{A},\mathbf{p}),
			\qquad \forall k\in\mathcal{N}.
			\label{eq:scaling_eint_rev}
		\end{align}
		Consequently, if $N^2\Sigma(\mathbf{A},\mathbf{p}) < \min_{k\in\mathcal{N}} \Theta_k$, 
		then the reduced objective in \eqref{eq:reduced_problem} strictly decreases under
		uniform power reduction until the boundary $N^2\Sigma(\mathbf{A},\mathbf{p}) = \min_{k\in\mathcal{N}} \Theta_k$ is reached or a lower power limit becomes active.
	\end{proposition}
	
	\begin{proof}
		For fixed \(\mathbf{A}\), the usable gains \(G_k(\mathbf{a}_k)\) are
		constants. Substituting \(p_j'=\alpha p_j\) into \eqref{eq:e_th_proxy} gives
		\begin{equation}
			\Sigma(\mathbf{A},\mathbf{p}')
			=
			\frac{\sigma_w^2}{N^2}
			\sum_{j=1}^{N}
			\frac{1}{\alpha p_j G_j(\mathbf{a}_j)}
			=
			\frac{1}{\alpha}\Sigma(\mathbf{A},\mathbf{p}),
		\end{equation}
		which proves \eqref{eq:scaling_sigma_rev}. Similarly, substituting
		\(p_k'=\alpha p_k\) and \(p_j'=\alpha p_j\) into the proxy bound in
		\eqref{eq:e_int_proxy} yields
		\begin{align}
			e_{\mathrm{int},k}(\mathbf{A},\mathbf{p}') 
			\leq & 
			\frac{1}{N^2 \alpha p_k G_k(\mathbf{a}_k)^2} \nonumber\\
			& \sum_{j\neq k}\alpha p_j
			\sum_{m=1}^{L}
			a_{m,k}\,\mathrm{tr}\!\bigl(\mathbf{\Phi}_{m,k}\mathbf{R}_{m,j}\bigr),  
		\end{align}
		so the common scaling factor cancels and the proxy value remains
		unchanged, establishing \eqref{eq:scaling_eint_rev}. The CSI term
		\(e_{\mathrm{csi},k}(\mathbf{A})\) in \eqref{eq:e_csi_proxy} is independent of power.
		
		Now suppose $N^2\Sigma(\mathbf{A},\mathbf{p}) < \min_{k\in\mathcal{N}} \Theta_k$ then  every maximum term in \eqref{eq:reduced_problem} is active, and the objective
		reduces to
		\begin{align}
			\sum_{k=1}^{N}\Theta_k
			&-
			(N-1)N^2\Sigma(\mathbf{A},\mathbf{p})\nonumber\\
			&+
			N^2\sum_{k=1}^{N}
			\bigl(
			e_{\mathrm{int},k}(\mathbf{A},\mathbf{p})
			+
			e_{\mathrm{csi},k}(\mathbf{A})
			\bigr). 
		\end{align}

		In \eqref{eq:scaling_sigma_rev}, decreasing \(\alpha\) increases
		\(\Sigma(\mathbf{A},\mathbf{p}')\), whereas the interference proxy term remains
		unchanged in \eqref{eq:scaling_eint_rev}. Hence the objective decreases strictly as long as the system
		remains in the regime where all privacy terms are active. The descent
		therefore continues until $N^2\Sigma(\mathbf{A},\mathbf{p}) = \min_{k\in\mathcal{N}} \Theta_k$ or until a lower power bound prevents further scaling.
	\end{proof}
	
	Proposition~\ref{prop:scaling_regime} formalizes the role of thermal
	noise in the reduced problem. In particular, it shows that, within the
	fully privacy-limited regime, uniform power reduction improves the
	objective without altering the interference proxy structure. Under
	homogeneous privacy targets, where \(\Theta_k=\Theta\) for all \(k\),
	the boundary reduces to \(N^2\Sigma(\mathbf{A},\mathbf{p})=\Theta\). This structural behavior heavily dictates the solution approach. Because the term $\Sigma(\mathbf{A}, \mathbf{p})$ couples the transmit powers of all UEs simultaneously within the objective, the power allocation subproblem does not decompose per-UE and must be solved jointly. 
	
	\begin{remark}
		\label{rem:homogeneous_privacy}
		In the homogeneous privacy scenario where all UEs share identical privacy targets ($\Theta_k = \Theta, \forall k$), the objective simplifies remarkably. The DP penalty collapses into a single threshold $N \cdot \max(0, \Theta - N^2\Sigma)$, yielding a pure system-wide state that tightly couples the privacy requirement with the aggregate thermal noise.
	\end{remark}
	
	\section{Proposed Cross-Layer Optimization Framework}
	\label{sec:solution}
	
	The reduced cross-layer design formulated in \eqref{eq:reduced_problem} is a mixed-integer non-convex program. The intractability of this formulation stems from the tight coupling between the combinatorial clustering variables $\mathbf{A}$ and the continuous transmit powers $\mathbf{p}$, exacerbated by the non-smooth maximum operator applied over the shared thermal noise constraint. To tackle this, we propose an alternating optimization (AO) framework that decouples the problem into two tractable blocks: a difference-of-convex (DC) power allocation subproblem and a greedy clustering update.
	
	\subsection{Power Allocation Subproblem}
	\label{subsec:power_alloc}
	
	In the first stage, we optimize the transmit powers $\mathbf{p}$ for a fixed clustering topology $\mathbf{A}$. For a fixed assignment, the effective gains $G_k(\mathbf{a}_k)$ and the CSI error proxies $e_{\mathrm{csi},k}(\mathbf{A})$ become deterministic constants. 
	
	We map the variables to the log-domain by defining $u_k = \ln p_k$. Consequently, the power bounds become $\ln P_{\min} \le u_k \le \ln P_{\max}$, where a positive lower bound $P_{\min} > 0$ reflects the minimum hardware transmit level. Under this transformation, the system-wide thermal noise and the scaled multi-UE interference proxies evaluate to
	\begin{equation}
		N^2\Sigma(\mathbf{u}) = \sigma_w^2 \sum_{j=1}^{N}\frac{e^{-u_j}}{G_j}, \quad N^2 e_{\mathrm{int},k}(\mathbf{u}) = \frac{1}{G_k^2}\sum_{j\neq k} e^{u_j - u_k}\,I_{kj},
		\label{eq:log_domain_terms}
	\end{equation}
	where $I_{kj} \triangleq \sum_{m=1}^{L} a_{m,k}\,\mathrm{tr}(\boldsymbol{\Phi}_{m,k}\mathbf{R}_{m,j}) \ge 0$. As both expressions are non-negative weighted sums of exponentials of affine arguments, they are convex functions of $\mathbf{u}$. 
	
	Substituting $\mathbf{u}$ into the objective of \eqref{eq:reduced_problem} yields the non-smooth function $\mathcal{U}(\mathbf{u}) = \sum_{k=1}^N \max(0, \Theta_k - N^2\Sigma(\mathbf{u})) + N^2\Sigma(\mathbf{u}) + N^2\sum_{k=1}^N e_{\mathrm{int},k}(\mathbf{u})$. However, the pointwise maximum of $0$ and the concave function $\Theta_k - N^2\Sigma(\mathbf{u})$ creates a non-convex geometry. To resolve this, we apply the algebraic identity $\max(0, \Theta_k - f) = \max(0, f - \Theta_k) - f + \Theta_k$. Setting $f = N^2\Sigma(\mathbf{u})$, the objective is exactly rearranged as
	\begin{align}
		\mathcal{U}(\mathbf{u}) =& \sum_{k=1}^N \max\!\left(0, N^2\Sigma(\mathbf{u}) - \Theta_k \right) + N^2\sum_{k=1}^N e_{\mathrm{int},k}(\mathbf{u})  \nonumber \\
		& + \sum_{k=1}^N \Theta_k - (N-1)N^2\Sigma(\mathbf{u}).
		\label{eq:dc_obj_exact}
	\end{align}
	
	This algebraic restructuring reveals a fundamental DC formulation, $\min_{\mathbf{u}} P(\mathbf{u}) - Q(\mathbf{u})$, defined precisely as
	\begin{align}
		P(\mathbf{u}) &= \sum_{k=1}^N \max\!\left(0, N^2\Sigma(\mathbf{u}) - \Theta_k\right) + N^2\sum_{k=1}^N e_{\mathrm{int},k}(\mathbf{u}), \label{eq:dc_p} \\
		Q(\mathbf{u}) &= (N-1)N^2\Sigma(\mathbf{u}). \label{eq:dc_q}
	\end{align}
	Because $\max(0, \cdot)$ is a convex non-decreasing function, its composition with the convex function $N^2\Sigma(\mathbf{u}) - \Theta_k$ is convex. Therefore, $P(\mathbf{u})$ constitutes a sum of convex functions. Similarly, $Q(\mathbf{u})$ is a positively weighted sum of convex exponentials, confirming its convexity.
	
	To optimize this DC structure, we employ the CCP. At each inner iteration $n$, we linearize the concave component $-Q(\mathbf{u})$ using a first-order Taylor expansion around the previous operating point $\mathbf{u}^{(n)}$. Evaluating the gradient yields the affine global overestimator, expressed as
	\begin{equation}
		-Q(\mathbf{u}) \approx -Q(\mathbf{u}^{(n)}) + (N-1)\sigma_w^2 \sum_{j=1}^N \frac{e^{-u_j^{(n)}}}{G_j} (u_j - u_j^{(n)}).
	\end{equation}
	To handle the non-smooth maximum operator in $P(\mathbf{u})$, we introduce auxiliary epigraph variables $\tau_k \ge 0$. Dropping the constant terms, the subproblem solved at CCP iteration $n$ evaluates to:
	\begin{subequations}
		\begin{align}
			\min_{\mathbf{u},\boldsymbol{\tau}} \quad 
			& \sum_{k=1}^{N}\tau_k + N^2\sum_{k=1}^{N} e_{\mathrm{int},k}(\mathbf{u}) + (N-1)\sigma_w^2 \sum_{j=1}^N \frac{e^{-u_j^{(n)}}}{G_j} u_j
			\label{eq:ccp_subprob}\\
			\text{s.t.}\quad
			& \tau_k \ge 0, \quad \forall k\in\mathcal{N}, \label{eq:ccp_con1} \\
			& \tau_k \ge N^2\Sigma(\mathbf{u}) - \Theta_k, \quad \forall k\in\mathcal{N}, \label{eq:ccp_con2} \\
			& \ln P_{\min} \le u_k \le \ln P_{\max}, \quad \forall k\in\mathcal{N}. \label{eq:ccp_con3}
		\end{align}
	\end{subequations}
	
	\begin{proposition}
		\label{prop:ccp_convex}
		Problem~\eqref{eq:ccp_subprob} is a convex program.
	\end{proposition}
	\begin{proof}
		The objective is linear in $\boldsymbol{\tau}$ and $\mathbf{u}$, plus a convex sum of affine-exponential interference terms $e_{\mathrm{int},k}(\mathbf{u})$. The non-trivial inequality \eqref{eq:ccp_con2} defines the region $\{(\tau_k,\mathbf{u}): \tau_k \ge N^2\Sigma(\mathbf{u}) - \Theta_k \}$, which is exactly the epigraph of the convex function $N^2\Sigma(\mathbf{u}) - \Theta_k$. The epigraph of a convex function forms a convex set. 
	\end{proof}
	It follows from Proposition~\ref{prop:ccp_convex} that the subproblem admits a global optimal solution computable in polynomial time via standard interior-point methods. Solving \eqref{eq:ccp_subprob} sequentially guarantees monotonic descent toward a stationary point of the non-convex power allocation space.
	
	\subsection{Clustering Subproblem}
	\label{subsec:clustering}
	
	In the second stage, we fix the transmit powers $\mathbf{p}$ generated by the CCP procedure and focus on optimizing the binary association matrix $\mathbf{A}$. The problem evaluates over a combinatorial set imposed by the $U_{\max}$ and $L_{\max}$ limits through constraints \eqref{eq:cluster_constraint} and \eqref{eq:load_constraint}. We adopt a greedy coordinate-descent strategy operating directly on the binary variables, avoiding fractional relaxations and rounding losses.
	
	Starting from the current clustering configuration $\mathbf{A}^{(i)}$, the procedure cycles through each UE $k \in \mathcal{N}$ and evaluates all feasible single-AP modifications to their serving set $\mathcal{L}_k$, where $\mathcal{L}_k \triangleq \{m \in \mathcal{L} : a_{m,k} = 1\}$:
	\begin{itemize}
		\item \emph{Addition:} For each AP $m \notin \mathcal{L}_k$ satisfying $|\mathcal{L}_k| < L_{\max}$ and $\sum_j a_{m,j} < U_{\max}$, the algorithm sets $a_{m,k} = 1$ and computes the resultant variation in the full surrogate objective \eqref{eq:reduced_problem}.
		\item \emph{Removal:} For each AP $m \in \mathcal{L}_k$ where $|\mathcal{L}_k| > 1$, the algorithm sets $a_{m,k} = 0$ and evaluates the objective impact.
	\end{itemize}
	If any modification reduces the utility surrogate, the steepest-descent candidate for UE $k$ is accepted before advancing to UE $k+1$. This pass iterates over $\mathcal{N}$ repeatedly until no UE accepts further structural modifications. 
	
	Each candidate evaluation incurs low computational overhead. Toggling a single link $a_{m,k}$ triggers an $\mathcal{O}(1)$ update to the usable gain $G_k$, directly translating into an $\mathcal{O}(1)$ adjustment to the global thermal noise $\Sigma$. Recalculating the interference proxy $e_{\mathrm{int},k}$ necessitates refreshing $N-1$ cross-UE trace terms, incurring a cost of $\mathcal{O}(N)$. In particular, under \eqref{eq:e_int_proxy}, toggling \(a_{m,k}\) affects only the victim-indexed interference term \(e_{\mathrm{int},k}\) through \(G_k(\mathbf{a}_k)\) and the summand indexed by \(k\), while \(e_{\mathrm{int},j}\) for \(j\neq k\) remains unchanged. Consequently, the aggregate algorithmic complexity per evaluation pass scales highly efficiently as $\mathcal{O}(N^2 L)$.
	
	\vspace{-2mm}
	\subsection{Algorithm Summary and Complexity Analysis}
	\label{subsec:algorithm}
	
	The proposed cross-layer protocol initializes the system with a proximity-based clustering matrix. To accelerate the CCP inner loop, the initial uniform transmit powers are explicitly calibrated to the privacy-noise regime boundary $p_k^{(0)} = \max\big(P_{\min}, \min\{\frac{\sigma_w^2 \sum_j G_j^{-1}}{\max_k \Theta_k},  P_{\max}\} \big), \forall k$. The whole procedure is detailed in Algorithm \ref{alg:ao}.
	
	\begin{algorithm}[t]
		\small
		\caption{Cross-Layer Alternating Optimization}
		\label{alg:ao}
		\SetAlgoLined
		\KwIn{$\{\mathbf{R}_{m,k},\mathbf{C}_{m,k}\}$, $\{\Theta_k\}$, $P_{\min}, P_{\max}, L_{\max}, U_{\max}$,  $\epsilon_{\mathrm{out}}, \epsilon_{\mathrm{in}}$.}
		\KwOut{$\mathbf{A}^\star, \mathbf{p}^\star, \boldsymbol{\sigma}^{2\star}$}
		Initialize $\mathbf{A}^{(0)}$ via proximity-based clustering\;
		Set $p_k^{(0)} = \max\big(P_{\min}, \min\{\frac{\sigma_w^2 \sum_j G_j^{-1}}{\max_k \Theta_k},  P_{\max}\} \big), \forall k$\;
		Set iteration index $i = 0$\;
		\Repeat{$|\mathcal{U}^{(i)} - \mathcal{U}^{(i-1)}| / \mathcal{U}^{(i-1)} < \epsilon_{\mathrm{out}}$}{
			\tcp{Block 1: Power Update (CCP)}
			Set $n = 0$ and $\mathbf{u}^{(n)} = \ln \mathbf{p}^{(i)}$\;
			\Repeat{$\|\mathbf{u}^{(n)} - \mathbf{u}^{(n-1)}\|_2 < \epsilon_{\mathrm{in}}$}{
				Solve convex program \eqref{eq:ccp_subprob}--\eqref{eq:ccp_con3} for fixed $\mathbf{A}^{(i)}$ and linearization point $\mathbf{u}^{(n)}$\;
				Update $\mathbf{u}^{(n+1)}$ with the optimal solution\;
				$n \leftarrow n + 1$\;
			}
			Set $\mathbf{p}^{(i+1)} = \exp(\mathbf{u}^{(n)})$\;
			\tcp{Block 2: Clustering Update}
			Run greedy coordinate-descent for fixed $\mathbf{p}^{(i+1)}$\;
			Set result to $\mathbf{A}^{(i+1)}$\;
			$i \leftarrow i + 1$\;
		}
		\tcp{Final DP Noise Recovery}
		Set $\sigma_k^{2\star} = \max\!\big(0, \Theta_k - N^2\Sigma(\mathbf{A}^\star, \mathbf{p}^\star)\big), \forall k$\;
	\end{algorithm}
	
	\begin{proposition}
		\label{prop:ao_convergence}
		The sequence of objective values 
		$\{\mathcal{U}^{(i)}\}_{i \ge 0}$ generated by 
		Algorithm~\ref{alg:ao} is monotonically non-increasing and 
		bounded below, hence convergent.
	\end{proposition}
	\begin{proof}
		The inner CCP loop monotonically descends the power 
		subproblem by standard DC programming theory, and the 
		coordinate-descent clustering accepts only modifications 
		that reduce or maintain the objective. The composite 
		sequence is therefore non-increasing and bounded below, 
		guaranteeing convergence.
	\end{proof}
	
	Any terminal point is blockwise locally optimal with respect 
	to the CCP power update and the greedy single-link 
	clustering modifications. However, because the problem 
	remains nonconvex with a combinatorial binary block, 
	Algorithm~\ref{alg:ao} is not guaranteed to reach a global 
	optimum of~\eqref{eq:prob_obj}; the solution should be 
	interpreted as a locally optimized operating point.
	
	\subsubsection{Computational Complexity}
	\label{rem:complexity}
	The per-iteration cost of Algorithm~\ref{alg:ao} is dominated by 
	two components. The CCP power allocation subproblem 
	in~\eqref{eq:ccp_subprob} involves $N$ optimization variables and 
	$\mathcal{O}(N)$ constraints, yielding a per-inner-iteration cost 
	of $\mathcal{O}(N^3)$ when solved via interior-point methods. The 
	greedy clustering pass evaluates at most $\mathcal{O}(NL)$ 
	candidate single-link modifications per sweep. The overall 
	complexity per AO iteration is therefore 
	$\mathcal{O}(I_{\mathrm{ccp}} N^3 + NL)$, where 
	$I_{\mathrm{ccp}}$ denotes the number of CCP inner iterations. In 
	practice, the AO framework converges within three to five outer 
	iterations across all tested configurations.
	
	\vspace{-2mm}
	\section{Numerical Analysis}
	\label{sec:results}
	\subsection{Simulation Setup}
	\label{sub:simulation}
	We consider a square coverage area of side length $D=500$~m, within which $L=36$ APs are deployed on a uniform $6\times6$ grid at height $10$~m. Each AP is equipped with $M=4$ antennas and serves a population of $N=20$ single-antenna UEs placed uniformly at random at height $1.5$~m. Wrap-around topology is applied to eliminate boundary effects \cite{ngo2017cell}. The carrier frequency is $f_c=1.9$~GHz, the system bandwidth is $B=20$~MHz, and a receiver noise figure of $9$~dB yields a per-complex-dimension thermal noise power of $\sigma_n^2 = k_{\mathrm{B}} T_0 B \cdot \mathrm{NF} \approx 8\times10^{-14}$~W, where $\mathrm{NF}$ is the noise figure in linear scale. Path loss follows the 3GPP UMi-NLOS model \cite{3gpp38901} with log-normal shadow fading of standard deviation $8$~dB, and the minimum distance is capped at $10$~m for model validity. UL channel estimation uses orthogonal pilot sequences of length $\tau_p = N$ at pilot transmit power $p_p = 200$~mW, and MMSE estimation under the isotropic spatial correlation model yields the per-antenna estimate and error covariances $\phi_{m,k}$ and $c_{m,k}$ defined in Section~\ref{subsec:clustering_channel}. The UC clustering budget allows each UE to be served by at most $L_{\max}=10$ APs, while each AP serves at most $U_{\max}=6$ UEs per round. Transmit powers are constrained to the interval $[P_{\min}, P_{\max}] = [1, 200]$~mW. For the privacy mechanism, the gradient clipping threshold is set to $C=1.0$ with failure probability $\delta=10^{-5}$, and the target privacy budget $\varepsilon$ is swept across a range of values from $0.5$ to $100$ to expose the regime structure identified in Section~\ref{sec:problem}. Unless stated otherwise, the alternating optimization is initialized with proximity-based clustering, and the inner (CCP) and outer (AO) convergence tolerances are set to $\epsilon_{\mathrm{in}}=10^{-5}$ and $\epsilon_{\mathrm{out}}=10^{-4}$, respectively. Unless stated otherwise, the nominal ergodic design proxy 
	$\kappa = 1$ is used in~\eqref{eq:robust_proxy}. To 
	quantify the impact of conservative privacy crediting, 
	Section~\ref{sec:hardening} additionally evaluates the 
	robust counterparts $\kappa \in \{0.5, 0.75\}$ at 
	$\varepsilon = 10$.
	\vspace{-2mm}
	\subsection{Regime Analysis and Optimization Behavior}
	\label{sub:regime_analysis}
	\subsubsection{Objective Decomposition Across Privacy Regimes}
	\label{subsub:regime_transition}
	Fig.~\ref{1.regime_transition} decomposes the optimized objective $\mathcal{U}^{\star}$ into its constituent terms across different privacy budget. The AO algorithm converges within three to five outer iterations for all tested values. For stringent privacy requirements ($\varepsilon \leq 5$), artificial DP noise constitutes the dominant component, accounting for 98.5\% of the objective at $\varepsilon=2$ and exceeding 99.9\% at $\varepsilon=0.5$. In this privacy-limited regime, thermal noise and multi-UE interference together contribute less than 1.5\%, confirming the analysis in Section~\ref{sec:problem} that physical-layer optimization yields marginal returns when the privacy threshold $\Theta_k$ far exceeds the achievable thermal noise floor $N^2\Sigma$.
	
	A sharp regime transition occurs between $\varepsilon=5$ and $\varepsilon=10$. At $\varepsilon=10$, the optimizer eliminates artificial DP noise entirely by driving $N^2\Sigma$ to the privacy boundary, as predicted by the reduced objective structure in~\eqref{eq:reduced_problem}. For all $\varepsilon \geq 10$, the objective is governed exclusively by physical-layer distortions, with multi-UE interference contributing 61\% to 71\% and thermal noise accounting for 29\% to 39\%. This ratio remains stable across the noise-limited range, indicating that the interference-thermal balance is determined by the network geometry and large-scale fading statistics rather than the privacy parameters. The CSI estimation error remains negligible throughout, below 0.05\%, consistent with the use of orthogonal pilots and the power-independent self-term structure derived in Appendix~\ref{app:csi_derivation}. The total optimized objective spans over four orders of magnitude from $\varepsilon=0.5$ to $\varepsilon=100$, underscoring the substantial performance improvement available as privacy constraints are relaxed.
	\begin{figure}[tb]
		\centering
		\includegraphics[scale=0.25]{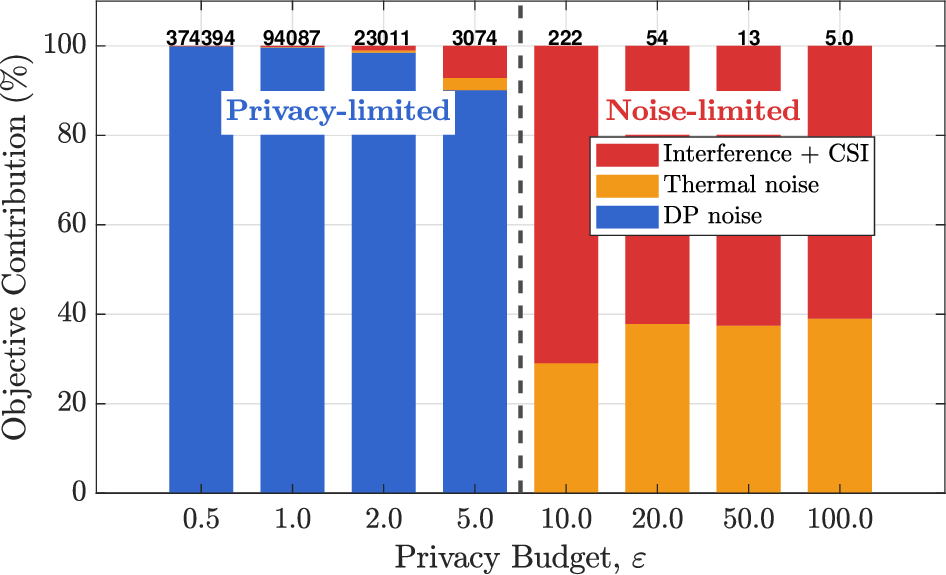}
		\caption{Optimized objective composition versus privacy budget $\epsilon$. The vertical dashed line marks the regime transition between privacy-limited and noise-limited operation.}
		\label{1.regime_transition}
	\end{figure}
	
	\subsubsection{Optimal Physical-Layer Configuration}
	\label{subsub:ap_power}
	Fig.~\ref{2.AP_power} reveals the physical-layer operating point selected by the optimizer as a function of $\varepsilon$. In the privacy-limited regime ($\varepsilon \leq 5$), all UEs are assigned exactly one serving AP and transmit at the minimum power $P_{\min}=1$~mW. This configuration is a direct consequence of the negative coefficient on $N^2\Sigma$ in the reduced objective identified in Section~\ref{sec:problem}: by minimizing the serving infrastructure and transmit power simultaneously, the optimizer maximizes the thermal noise floor, thereby substituting physical noise for costly artificial DP noise. The resulting transmit power is nearly 23~dB below $P_{\max}$, yielding a substantial reduction in UE energy expenditure while maintaining the target ergodic privacy design constraint.
	
	As the privacy budget increases beyond $\varepsilon=10$, the threshold $\Theta_k$ drops below $N^2\Sigma$ and the optimizer transitions to improving aggregation fidelity. The mean number of serving APs grows from 1.5 at $\varepsilon=10$ to 5.8 at $\varepsilon=100$, reflecting the increasing value of AP macro-diversity for suppressing multi-UE interference, which accounts for over 60\% of the objective in this regime as shown in Fig.~\ref{1.regime_transition}. Transmit power increases only modestly to 1.37~mW at $\varepsilon=50$, and in fact decreases slightly to 1.32~mW at $\varepsilon=100$, indicating that the optimizer continues to favor clustering expansion over power escalation. This behavior confirms that under the CF architecture considered here, spatial diversity through additional AP associations is more effective at reducing aggregation variance than increasing transmit power, which would simultaneously amplify cross-UE interference through the coupling structure in~\eqref{eq:e_int_proxy}. In the noise-limited regime, increasing the AP association budget from $L_{\max} = 1$ to $L_{\max} = 5$ captures over $85\%$ of the achievable clustering gain, with diminishing returns beyond $L_{\max} = 5$. In the privacy-limited 
	regime, the optimizer selects a single serving AP per UE 
	regardless of the available budget.
	\begin{figure}[tb]
		\centering
		\includegraphics[scale=0.25]{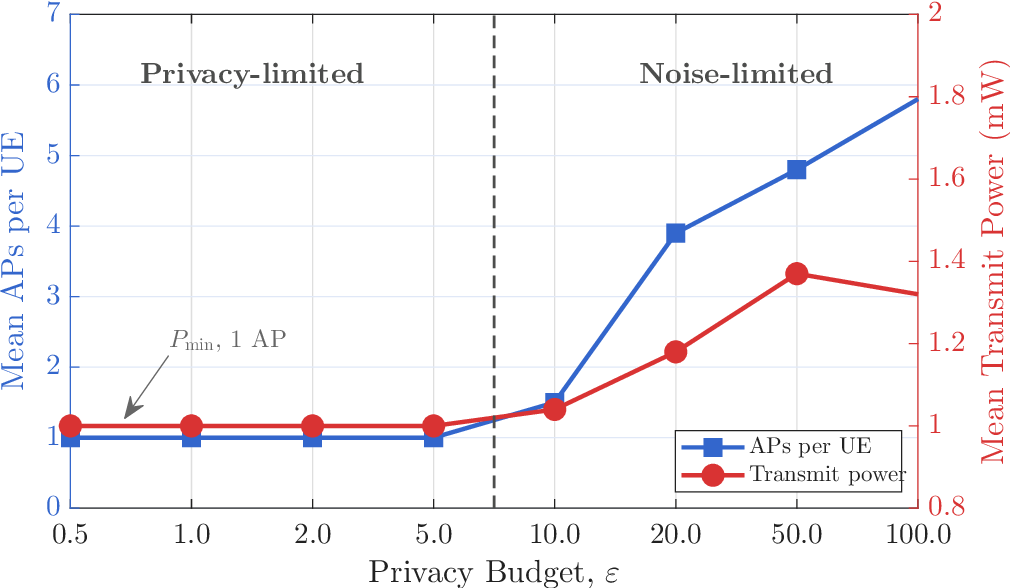}
		\caption{Optimized mean AP association and transmit power versus privacy budget $\epsilon$. The dashed line separates the privacy-limited and noise-limited operating regimes.}
		\label{2.AP_power}
	\end{figure}
	
	\subsubsection{Privacy-Noise Boundary Targeting}
	\label{subsub:Clustering_Topology}
	Fig.~\ref{7.Clustering_Topology} directly verifies the boundary-targeting mechanism predicted by the reduced objective analysis in Section~\ref{sec:problem}. The privacy threshold $\Theta_k$ decreases monotonically with $\varepsilon$, spanning four orders of magnitude from $\varepsilon=0.5$ to $\varepsilon=100$. In the privacy-limited regime ($\varepsilon \leq 5$), the optimized thermal noise floor $N^2\Sigma$ remains well below $\Theta_k$, with a gap exceeding two orders of magnitude at $\varepsilon=0.5$. Despite the minimal clustering and power allocation observed in Fig.~\ref{2.AP_power}, the physical-layer degrees of freedom are insufficient to close this gap, and the residual privacy cost is borne by artificial DP noise. We emphasize that this boundary alignment is an ergodic design property: the optimizer targets the expected thermal noise floor $N^2\Sigma = \Theta_k$, with instantaneous fluctuations governed by the channel-hardening concentration characterized in Remark \ref{rem:privacy_scope}.
	
	For $\varepsilon \geq 10$, the two curves converge and remain tightly matched. At $\varepsilon=10$, the optimizer achieves $N^2\Sigma=\Theta_k$, eliminating artificial noise entirely. This alignment persists at $\varepsilon=20$ and $\varepsilon=50$, confirming that the optimizer consistently drives the system to the regime boundary $N^2\Sigma=\Theta_k$ where the substitution effect is fully exploited. Beyond this boundary, additional degrees of freedom are allocated to reducing interference through AP diversity, as reflected in the clustering expansion observed in Fig.~\ref{2.AP_power}. This convergence validates the structural insight that the optimal cross-layer design treats the physical thermal noise floor as a free privacy-masking resource, activating artificial noise injection only when the privacy demand exceeds what the channel environment can provide.
	\begin{figure}[tb]
		\centering
		\includegraphics[scale=0.25]{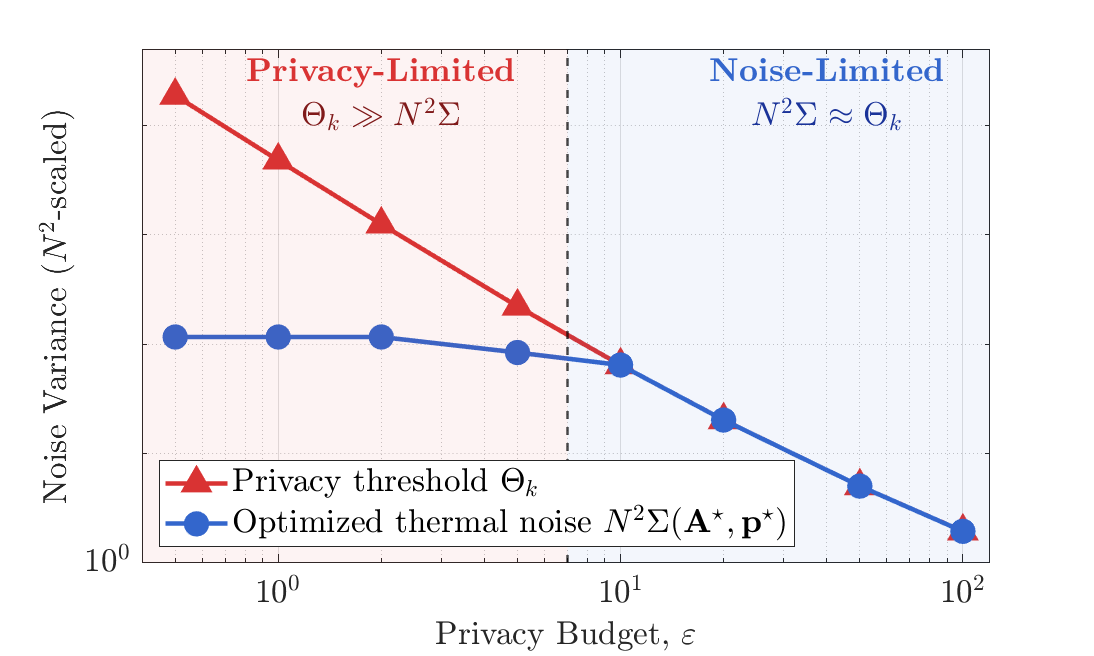}
		\caption{Regime boundary targeting. The optimizer drives $N^2\Sigma^{\star}$ 
			toward the privacy threshold $\Theta$ as $\varepsilon$ increases. In the 
			privacy-limited regime ($\varepsilon \leq 5$), thermal noise remains below 
			the threshold, requiring artificial DP noise. At the boundary 
			($\varepsilon \approx 10$), $N^2\Sigma^{\star} \approx \Theta$, eliminating 
			artificial noise entirely. In the noise-limited regime ($\varepsilon \geq 20$), 
			physical noise alone satisfies privacy.}
		\label{fig:boundary_targeting}
		\label{7.Clustering_Topology}
	\end{figure}
	\vspace{-3mm}
	\subsection{Channel-Hardening Quality and Design Robustness}
	\label{sec:hardening}
	
	The ergodic privacy proxy developed in Section~\ref{subsec:zcdp} relies on the concentration of the instantaneous self-gain 
	$S_k = \sum_{m=1}^{L} a_{m,k}\|\hat{\mathbf{h}}_{m,k}\|^2$ 
	around its mean~$G_k(\mathbf{a}_k)$. We now quantify this 
	concentration for the simulated deployment and evaluate the 
	sensitivity of the optimized objective to the robustness 
	parameter~$\kappa$ introduced in~\eqref{eq:robust_proxy}.
	
	\subsubsection{Hardening Concentration}
	
	Under the isotropic spatial correlation model with $M = 4$ 
	antennas per AP, Table~\ref{tab:hardening} reports the hardening quality metric~$\beta_k$ from~\eqref{eq:beta_k}, the coefficient of variation $\mathrm{CV}(S_k) = 1/\sqrt{\beta_k}$, and the resulting 
	Chebyshev failure bounds across all $N = 20$~UEs under the 
	two clustering configurations identified in 
	Section~\ref{subsub:ap_power}. Because the Chebyshev bound in Remark~\ref{rem:sensitivity} is distribution-free, it is naturally conservative for the chi-squared-type self-gain variables arising here and should therefore be interpreted as a worst-case concentration certificate rather than a tight characterization of the instantaneous privacy shielding. Table~\ref{tab:hardening} reveals that in the privacy-limited regime where the optimizer selects a single serving AP per UE (Section~\ref{subsub:ap_power}), the hardening metric is pinned at its minimum $\beta_k = M = 4$, and both the per-round and cumulative Chebyshev bounds exceed unity. 
	These bounds are therefore vacuous and provide no meaningful 
	concentration certificate at this operating point. Additionally, 
	expanding the serving set in the noise-limited regime 
	improves concentration for well-positioned UEs but leaves 
	the worst-positioned UE at $\beta_k = 4.0$, indicating that 
	a single poorly connected UE can limit the system-wide 
	hardening quality regardless of the clustering budget.
	
	\begin{table}[t]
		\centering
		\caption{Channel-hardening quality under optimized clustering
			for two operating regimes ($M=4$, $N=20$, $\xi=0.5$).}
		\label{tab:hardening}
		\footnotesize
		\setlength{\tabcolsep}{3.5pt}
		\renewcommand{\arraystretch}{1.08}
		\begin{tabular}{@{}lcccc@{}}
			\toprule
			\textbf{Regime} & \textbf{$\beta_k$} & \textbf{CV$(S_k)$} & \textbf{Per-round}$^{\dagger}$ & \textbf{Cumul.}$^{\ddagger}$ \\
			\midrule
			\shortstack[l]{Single-AP\\($\varepsilon \le 5$)}
			& $4.0$
			& $50\%$
			& $1.00$
			& $200$ \\
			
			\shortstack[l]{Multi-AP\\($\varepsilon \ge 10$)}
			& $[4.0,\,22.7]$
			& $[21\%,\,50\%]$
			& $[0.18,\,1.00]$
			& $[35.2,\,200]$ \\
			\bottomrule
		\end{tabular}
		
		\vspace{1mm}
		\parbox{\columnwidth}{\scriptsize
			Ranges are $[\text{best UE},\,\text{worst UE}]$.
			$^{\dagger}$ Per-round Chebyshev bound: $1/(\xi^2\beta_k)$.
			$^{\ddagger}$ Cumulative union bound over $T=200$ rounds: $T/(\xi^2\beta_k)$.
			Bounds exceeding $1$ are formally vacuous.
		}
	\end{table}
	
	These results confirm that, for $M = 4$, the Chebyshev-based 
	bound from Remark~\ref{rem:sensitivity} yields conservative bounds that do not provide formal realization-wise certificates at this antenna count. Since $\beta_k \geq M$ with equality for single-AP 
	service under isotropic correlation, deployments with larger 
	antenna arrays directly strengthen the concentration and 
	narrow the gap between the ergodic proxy and instantaneous 
	behavior. For the $M = 4$ configuration studied here, the 
	formal gap motivates the use of the robust 
	privacy requirement~\eqref{eq:robust_constraint} with $\kappa < 1$, 
	which tightens the privacy calibration through a controlled 
	reduction of the credited thermal noise, independently of 
	the concentration bound quality. For this reason, we next evaluate the robust counterpart 
	in~\eqref{eq:robust_constraint}, which tightens the credited thermal-noise 
	contribution independently of the concentration-bound 
	looseness.
	
	\subsubsection{Impact of Conservative Privacy Crediting}
	
	To quantify the cost of conservative privacy crediting, 
	we solve Problem~\eqref{eq:reduced_problem} at 
	$\varepsilon = 10$ for $\kappa \in \{1.0, 0.75, 0.50\}$. 
	Table~\ref{tab:kappa} reports the optimized objective, the 
	mean per-UE artificial noise, and the relative objective 
	increase for each setting.
	
	\begin{table}[t]
		\centering
		\caption{Impact of robustness parameter~$\kappa$ on the 
			optimized objective at $\varepsilon = 10$.}
		\label{tab:kappa}
		\footnotesize
		\renewcommand{\arraystretch}{1.15}
		\begin{tabular}{@{} c c c c @{}}
			\toprule
			$\kappa$ 
			& $\bar{\sigma}_k^{2\star}$ 
			& $\mathcal{U}^\star$ 
			& $\Delta\mathcal{U}$ \\
			\midrule
			$1.00$ & $\approx 0$  & $222.0$ & (baseline) \\
			$0.75$ & $16.1$       & $544.3$ & $+145.4\%$ \\
			$0.50$ & $32.3$       & $866.8$ & $+290.8\%$ \\
			\bottomrule
		\end{tabular}
		
		\vspace{1mm}
		\raggedright
		\scriptsize
		$\bar{\sigma}_k^{2\star}$: mean per-UE artificial noise.\;
		$\Delta\mathcal{U}$: relative increase over $\kappa = 1$.
	\end{table}
	
	Table~\ref{tab:kappa} shows that the objective increase is almost entirely attributable to re-injected artificial noise rather than changes in the physical-layer configuration. At $\kappa = 0.75$, the total artificial noise across all UEs, which accounts for over $99\%$ of 
	the objective increase relative to $\kappa = 1$. 
	The same pattern holds at $\kappa = 0.50$. This indicates 
	that the robust formulation adjusts only the noise allocation 
	while the underlying power levels and clustering topology 
	remain essentially unchanged, confirming that the robust counterpart operates as a pure noise-allocation adjustment without disrupting the optimized physical-layer operating point.
	
	Two properties of the conservative design are worth noting. 
	First, Algorithm~\ref{alg:ao} converges within the same 
	number of outer iterations for all three~$\kappa$ values; 
	only the privacy constraint 
	boundary~\eqref{eq:robust_constraint} shifts, preserving 
	the computational structure of the solver. Second, the proposed framework at $\kappa = 1.0$ with 
	$\varepsilon = 10$ achieves reductions exceeding $82\%$ relative to all 
	baselines (Section~\ref{subsec:baseline_comparison}). The choice of~$\kappa$ in practice should be guided by the 
	deployment-specific fading statistics and the required 
	confidence level for the privacy calibration. The results 
	above demonstrate that the framework maintains structural 
	advantages over conventional baselines across the full tested 
	range: even at $\kappa = 0.50$, the optimized objective 
	$\mathcal{U}^\star = 866.8$ remains below all baseline 
	values reported in Section~\ref{subsec:baseline_comparison} at 
	$\varepsilon = 10$, indicating that the physical-layer operating point optimized at $\kappa = 1$ remains structurally valid under tightened privacy crediting, with only the noise-allocation layer absorbing the additional conservatism.
	
	\vspace{-2mm}
	\subsection{Baseline Comparison}
	\label{subsec:baseline_comparison}
	
	To evaluate the contribution of each design component, we compare the 
	proposed cross-layer framework against five baselines. 
	\emph{Privacy-Unaware} runs the full AO with $\Theta_k \approx 0$ 
	and evaluates the result under the actual $\varepsilon$-constraint, 
	quantifying the cost of ignoring privacy during physical-layer 
	design~\cite{sifaou2024ota}. \emph{Fixed-Cluster} retains the 
	initial proximity-based AP association and optimizes only transmit 
	powers via CCP under the privacy-aware objective, serving as the 
	CF counterpart to single-cell DP-aware power-control 
	designs~\cite{liu2021privacy, jiang2023optimized, yan2024over}. 
	\emph{Equal-Power} fixes all UEs at $P_{\max} = 200$~mW and 
	optimizes only clustering, isolating the benefit of privacy-aware 
	power control. \emph{Single-Cell} assigns each UE to its single 
	strongest AP at $P_{\max}$, representing conventional cellular 
	operation without macro-diversity. \emph{Naive-DP} uses proximity 
	clustering at $P_{\max}$ with the full DP budget 
	$\sigma_k^2 = \Theta_k$ injected locally, without leveraging 
	thermal noise as a privacy resource.
	
	Fig.~\ref{4.objective} plots the resulting objective $\mathcal{U}^\star$ 
	on a logarithmic scale for $\varepsilon \in \{1, 2, 5, 10, 20, 50\}$. 
	The proposed scheme achieves the lowest objective across all tested 
	values. At $\varepsilon = 1$, all six schemes produce comparable 
	objectives within 2.1\% of each other, since the DP noise floor 
	$\Theta_k \approx 4.8 \times 10^3$ overwhelms all physical-layer 
	contributions as discussed in Section~\ref{subsub:regime_transition}. 
	In this deeply privacy-limited regime, the dominant design action is 
	power reduction rather than sophisticated resource allocation. As 
	$\varepsilon$ increases, the gains of joint optimization become 
	pronounced. At $\varepsilon = 10$, the proposed scheme reduces the 
	objective by 82.8\% over Privacy-Unaware, 82.8\% over Fixed-Cluster, 
	and 83.1\% over Naive-DP. This improvement originates from the 
	boundary-targeting mechanism verified in 
	Section~\ref{subsub:Clustering_Topology}, which eliminates artificial 
	noise entirely while conventional baselines continue to inject 
	$\sigma_k^2 \approx \Theta_k$. The gap widens further at 
	$\varepsilon = 50$, where the proposed scheme achieves reductions of 86.8\% and 
	89.1\% compared to Privacy-Unaware 
	and Naive-DP, respectively.  Among the baselines, Equal-Power and Naive-DP 
	consistently produce the highest objectives due to their use of 
	$P_{\max}$, which minimizes the thermal noise floor and forfeits 
	the substitution effect. The Fixed-Cluster and 
	Privacy-Unaware schemes perform comparably to each other, suggesting 
	that when clustering is either fixed at the proximity solution or 
	optimized without privacy awareness, the resulting AP associations 
	are similarly suboptimal under the actual privacy constraint. These 
	results confirm that the performance gains of the proposed framework 
	are driven primarily by the joint exploitation of the privacy-noise 
	substitution and adaptive clustering, neither of which is captured 
	by any single baseline in isolation.
	\begin{figure}[tb]
		\centering
		\includegraphics[scale=0.25]{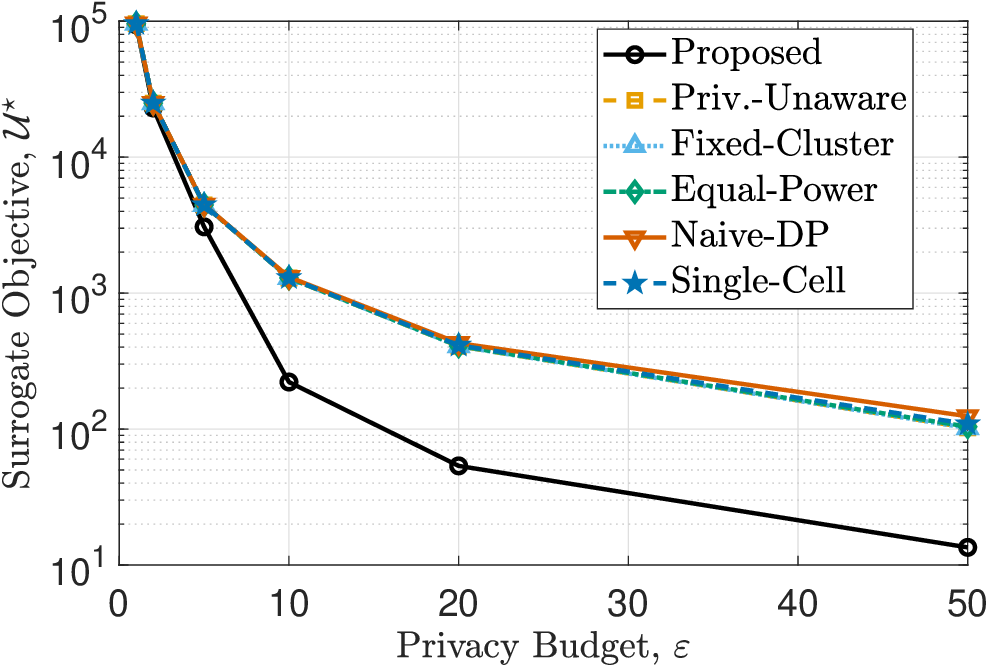}
		\caption{Impact of privacy budget $\epsilon$ on the optimized objective $\mathcal{U}^\star$ for all baseline schemes.}
		\label{4.objective}
	\end{figure}
	
	We note that the baselines are designed to isolate individual 
	design components rather than replicate a specific prior method. 
	Since no existing work jointly optimizes UC clustering, 
	CF OTA power control, and per-UE DP 
	noise under the architecture considered here, a direct one-to-one 
	adaptation of prior single-cell or digital-FL frameworks is not 
	straightforward. In particular, single-cell DP-OTA power control 
	strategies such as~\cite{wei2022low, jiang2023optimized} rely on a 
	common alignment coefficient across all UEs, which does not 
	translate to the distributed per-AP combining and scalar 
	forwarding model in~\eqref{eq:rx_ap}. 
	
	\subsection{FL Training Validation}
	\label{subsec:fl_validation}
	
	To validate that the optimized physical-layer configurations translate 
	into tangible learning performance gains, we conduct E2E 
	federated learning simulations using the FedAvg protocol. Each UE 
	computes a local model-delta update through $\tau = 1$ mini-batch stochastic gradient descent (SGD) 
	step per communication round, clips the update to $\ell_2$-norm bound 
	$C = 1.0$, and injects Gaussian DP noise with variance 
	$\sigma_k^{2\star}$ determined by the optimized cross-layer 
	configuration. The privatized updates are aggregated at the CPU 
	using the statistical OTA noise model, where the total physical-layer 
	distortion $\Sigma + \sum_k (e_{\mathrm{int},k} + e_{\mathrm{csi},k})$ 
	is injected as additive Gaussian noise consistent with the variance 
	decomposition in~\eqref{eq:ota_decomp}. Training data is partitioned 
	across $N = 20$ UEs using a Dirichlet distribution with concentration 
	parameter $\alpha_{\text{Dir}} = 0.5$, yielding approximately 3{,}000 samples per 
	UE for MNIST~\cite{lecun1998mnist}  and 2{,}500 for CIFAR-10~\cite{krizhevsky2009cifar}. We evaluate on MNIST using a two-layer convolutional neural network (CNN) with 
	$d = 80{,}202$ parameters and on CIFAR-10 
	using a similar architecture with $d = 1{,}070{,}794$ parameters. The 
	complete set of FL hyperparameters is summarized in 
	Table~\ref{tab:fl_hyper}. The \emph{Proposed} cross-layer design is 
	compared with the following three baselines: \emph{Naive-DP} 
	(proximity clustering, $P_{\max}$, full DP budget $\sigma_k^2 = \Theta_k$); 
	\emph{Privacy-Unaware} (AO-optimized for $\varepsilon \to \infty$, 
	evaluated under the actual privacy constraint); and \emph{Fixed-Cluster}. The Fixed-Cluster baseline, already introduced 
	in Section~\ref{subsec:baseline_comparison}, serves here as the 
	privacy-aware fixed-topology benchmark, thereby isolating the 
	contribution of clustering adaptation in the E2E FL setting.
	
	\begin{table}[t]
		\centering
		\caption{Federated learning hyperparameters}
		\label{tab:fl_hyper}
		\begin{tabular}{lcc}
			\toprule
			Parameter & Symbol & Value \\
			\midrule
			Communication rounds & $T$ & 200 \\
			Local SGD steps per round & $\tau$ & 1 \\
			Local learning rate & $\eta$ & 0.1 \\
			Global step size & $\eta_g$ & 1.0 \\
			Mini-batch size & $B$ & 32 \\
			Gradient clipping bound & $C$ & 1.0 \\
			Dirichlet concentration & $\alpha_{\text{Dir}}$ & 0.5 \\
			OTA channel mode & -- & Statistical \\
			Evaluation interval & -- & 10 rounds \\
			\bottomrule
		\end{tabular}
	\end{table}
	
	Fig.~\ref{MNIST_10} presents the MNIST test accuracy trajectory at $\varepsilon=10$, where the Proposed scheme converges reliably to 96.89\% after $T=200$ rounds, while all baselines remain near 12.7\%. At $\varepsilon=50$, the Proposed scheme reaches 97.05\% whereas the baselines plateau near 22.96\% 
	(Table~\ref{tab:fl_accuracy}), confirming the pattern observed at 
	$\varepsilon = 10$. This 
	performance gap confirms that the dominant E2E gain arises from the 
	joint optimization of clustering and power, rather than from 
	privacy-aware power control alone. In particular, the near-overlap of 
	Fixed-Cluster with Naive-DP and Privacy-Unaware indicates that, under a 
	fixed topology, power adaptation by itself is insufficient to move the 
	system into a favorable learning regime. The Proposed scheme, on the 
	other hand, operates at $\sigma_k^{2\star} \approx 0$ by steering the 
	system toward the boundary $N^2\Sigma = \Theta_k$, thereby replacing 
	costly artificial DP noise with physical-layer noise while preserving 
	a sufficiently small aggregation distortion for the MNIST task.
	
	\begin{figure}[tb]
		\centering
		\includegraphics[scale=0.25]{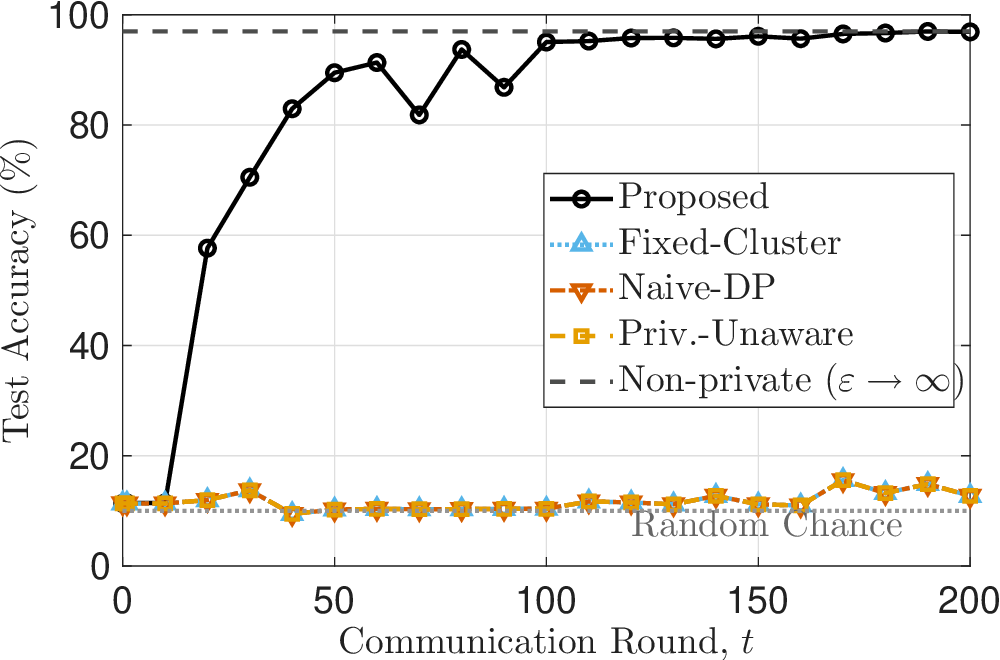}
		\caption{Test accuracy versus communication round on MNIST at 
			$\varepsilon = 10$. The dashed line indicates the 
			non-private upper bound ($\varepsilon \to \infty$).}
		\label{MNIST_10}
	\end{figure}
	
	Fig.~\ref{CIFAR10_50} extends the evaluation to CIFAR-10 at 
	$\varepsilon = 50$ with $d = 1{,}070{,}794$ parameters. The Proposed 
	scheme reaches 47.71\% accuracy with a clearly rising trend at round 
	200, whereas Fixed-Cluster, Naive-DP, and Privacy-Unaware remain near 
	random chance at 10.78\%, 10.74\%, and 10.74\%, respectively. The 
	lower absolute accuracy relative to MNIST reflects the increased 
	difficulty of CIFAR-10 under a lightweight two-layer CNN with 
	single-batch local updates, rather than a limitation of the cross-layer 
	framework. At the same time, the 36.93 percentage point 
	separation from the strongest baseline confirms that the proposed 
	boundary-targeting mechanism remains effective even for a model with 
	roughly 13$\times$ more parameters than the MNIST architecture. 
	
	\begin{figure}[tb]
		\centering
		\includegraphics[scale=0.25]{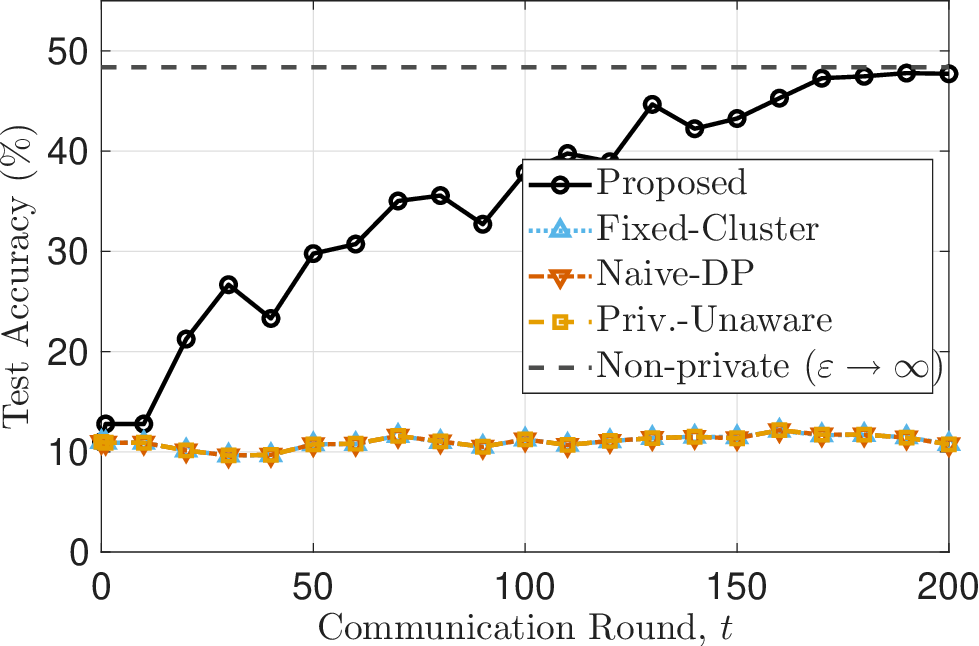}
		\caption{Test accuracy versus communication round on CIFAR-10 at 
			$\varepsilon = 50$. The dashed line indicates the 
			non-private upper bound ($\varepsilon \to \infty$).}
		\label{CIFAR10_50}
	\end{figure}
	
	Table~\ref{tab:fl_accuracy} consolidates the final test accuracy 
	across all configurations. A consistent pattern emerges: the Proposed 
	scheme enables meaningful learning whenever the privacy budget is 
	sufficient to sustain a feasible noise-to-signal ratio (NSR), whereas the 
	other three baselines remain in poor or infeasible regimes. At 
	$\varepsilon = 2$, all four schemes produce identical accuracy of 
	11.36\%, confirming the infeasibility predicted by the NSR analysis. With $\Theta_k \approx 1.25 \times 10^3$ at this 
	privacy level, even the Proposed scheme still requires 
	$\sigma_k^{2\star} = 1.13 \times 10^3$, resulting in a per-coordinate 
	NSR exceeding $10^6$ and preventing any gradient signal from surviving 
	the aggregation process. This marks the boundary between feasible and 
	infeasible learning under the considered system constraints. For 
	$\varepsilon \in \{10,50\}$, however, the Proposed method is able to 
	push the system into a noise-limited regime in which 
	$\sigma_k^{2\star} \approx 0$, thereby restoring near-noiseless learning. 
	By contrast, Fixed-Cluster, Naive-DP, and Privacy-Unaware remain unable 
	to exploit this regime transition effectively, which explains their 
	persistent collapse in E2E FL accuracy. These results therefore 
	identify the privacy range over which cross-layer optimization is 
	actionable and demonstrate that clustering adaptation is the primary driver of  converting physical-layer privacy shielding into actual learning 
	gains.
	\begin{table}[t]
		\centering
		\caption{E2E FL test accuracy (\%) under DP}
		\label{tab:fl_accuracy}
		\resizebox{\columnwidth}{!}{%
			\begin{tabular}{llccccc}
				\toprule
				Dataset & $\varepsilon$ & Proposed & Fixed-Clust. & Naive-DP & Priv.-Unaware & Gap$^\dagger$ \\
				\midrule
				\multirow{3}{*}{MNIST}
				& 2  & 11.36 & 11.36 & 11.36 & 11.36 & 0.00 \\
				& 10 & \textbf{96.89} & 12.73 & 12.72 & 12.72 & 84.16 \\
				& 50 & \textbf{97.05} & 22.96 & 22.72 & 22.74 & 74.09 \\
				\midrule
				CIFAR-10
				& 50 & \textbf{47.71} & 10.78 & 10.74 & 10.74 & 36.93 \\
				\bottomrule
				\multicolumn{7}{l}{\footnotesize $^\dagger$Gap = Proposed $-$ best baseline (percentage points).} \\
				\multicolumn{7}{l}{\footnotesize Random chance: 10\%. At $\varepsilon{=}2$, all methods are infeasible (NSR $>10^6$).}
			\end{tabular}%
		}
	\end{table}
	
	The FL experiments reported above use a single network topology 
	realization and random seed, with the OTA channel modeled through 
	the statistical variance proxy rather than per-round fading 
	realizations. This setup isolates the impact of the cross-layer 
	optimization on the aggregation noise floor while maintaining 
	computational tractability across the full configuration sweep. 
	The lightweight CNN architectures and single-batch local updates 
	are chosen to make the physical-layer configuration the primary 
	controlled performance differentiator rather than the model 
	capacity or local computation budget. Evaluation under multiple 
	topology realizations, per-round fading channels, and larger model 
	architectures remains a direction for future investigation.
	\vspace{-2mm}
	\section{Conclusion}
	\label{sec:conclusion}
	
	In this paper, we establish that the thermal noise floor in CF mMIMO 
	ULs can serve as a controllable privacy-design resource for OTA FL under ergodic channel-hardening conditions. By 
	formulating a cross-layer optimization that couples AP-UE clustering, 
	transmit power allocation, and artificial noise injection under 
	per-UE ergodic zCDP design constraints, the proposed framework 
	exposes a regime-dependent operating structure that governs the 
	interaction between privacy requirements and physical-layer 
	configuration. In the privacy-limited regime, the optimizer 
	deliberately reduces the utilization of AP-UE links and transmit power, 
	substituting inherent channel noise for costly artificial 
	perturbation. Beyond a critical privacy threshold, the framework 
	transitions to noise-limited operation, where artificial noise 
	vanishes entirely and all optimization effort is redirected toward 
	suppressing multi-UE interference through spatial diversity. The 
	boundary-targeting mechanism that mediates this transition ensures 
	that the system operates at the point where the physical noise floor precisely satisfies the ergodic privacy design constraint in expectation, eliminating wasteful over-provisioning of artificial noise. For finite antenna arrays where channel-hardening 
	concentration is moderate, the robust formulation provides 
	a tunable mechanism to strengthen the privacy calibration 
	at a controlled increase in aggregation variance, allowing 
	operators to balance thermal-noise credit against design 
	conservatism based on deployment-specific conditions. E2E training 
	experiments on both MNIST and CIFAR-10 confirm that this variance 
	reduction translates into a qualitative shift in learning outcomes, 
	enabling convergent training under privacy budgets where conventional 
	approaches fail to exceed random-chance accuracy. The framework is agnostic to the learning task and model architecture, requiring only large-scale fading statistics and privacy targets as inputs. These quantities vary slowly relative to the FL training timescale, enabling the proposed optimization to be executed offline 
	or periodically, with limited additional overhead relative to the FL training timescale in practical CF deployments. 
	Future work includes extending the analysis to multi-antenna UEs and frequency-selective channels, incorporating adaptive 
	clipping mechanisms that jointly optimize the sensitivity-bias 
	tradeoff, and investigating the interaction between partial UE 
	participation and the identified regime structure under practical 
	fronthaul capacity constraints.
	
	\vspace{-2mm}
	\appendices
	\section{Derivation of the CSI Error Proxy}
	\label{app:csi_derivation}
	
	We derive the variance proxy for the CSI error and demonstrate that its dominant self-term is independent of the transmit power $p_k$. Based on the OTA decomposition, the residual distortion caused by imperfect channel estimates at the CPU is $e_{\mathrm{csi}}^{(t)} = \sum_{k=1}^N \sum_{m=1}^{L} \mathbf{v}_m^{(t)H} \tilde{\mathbf{h}}_{m,k}^{(t)} x_k^{(t)}$.
	Assuming the normalized model updates have unit variance, we have $\mathbb{E}[|x_k|^2] = p_k$. Since the MMSE estimation error $\tilde{\mathbf h}_{m,k}$ is zero-mean and independent across APs, the variance contribution associated with UE~$k$ evaluates to $\mathrm{Var}\!\left( e_{\mathrm{csi},k} \right) = p_k \sum_{m=1}^{L} \mathbb{E}\!\left[ \mathbf{v}_m^H \mathbf{C}_{m,k} \mathbf{v}_m \right]$. Substituting the local combiner $\mathbf{v}_m = \sum_{j=1}^{N} a_{m,j} c_j \hat{\mathbf{h}}_{m,j}$ yields:
	\begin{align}
		\mathrm{Var}\!\left( e_{\mathrm{csi},k} \right) = p_k \sum_{m=1}^{L} & \sum_{j=1}^{N} \sum_{j'=1}^{N} a_{m,j}  a_{m,j'} c_j c_{j'} \nonumber\\ 
		& \mathbb{E}\!\left[ \hat{\mathbf{h}}_{m,j}^H \mathbf{C}_{m,k} \hat{\mathbf{h}}_{m,j'} \right].
		\label{eq:app_csi_var_step2}
	\end{align}
	
	In mMIMO systems, the dominant contribution in~\eqref{eq:app_csi_var_step2} arises from the self-aligned term $j=j'=k$, which yields $\mathrm{Var}_{\mathrm{self}}\!\left( e_{\mathrm{csi},k} \right) 
	= p_k \sum_{m=1}^{L} a_{m,k} c_k^2 \, \mathrm{tr}\!\big( \boldsymbol{\Phi}_{m,k} \mathbf{C}_{m,k} \big)$. Substituting the squared equalization coefficient $c_k^2 = \frac{1}{N^2 p_k G_k(\mathbf{a}_k)^2}$ reveals a perfect algebraic cancellation of the transmit power $p_k$, written as
	\begin{equation}
		\mathrm{Var}_{\mathrm{self}}\!\left( e_{\mathrm{csi},k} \right) = \frac{1}{N^2 G_k(\mathbf{a}_k)^2} \sum_{m=1}^{L} a_{m,k} \, \mathrm{tr}\!\big( \boldsymbol{\Phi}_{m,k} \mathbf{C}_{m,k} \big).
		\label{eq:app_csi_var_final}
	\end{equation}
	
	The remaining terms in \eqref{eq:app_csi_var_step2} correspond to cross-UE CSI-leakage contributions. When the channel estimates are mutually uncorrelated across UEs, as in the orthogonal-pilot setting with $\tau_p=N$, the off-diagonal terms $j\neq j'$ vanish in expectation, while the residual diagonal cross-UE terms $j=j'\neq k$ are proportional to $a_{m,j}c_j^2 p_k\,\mathrm{tr}\!\bigl(\boldsymbol{\Phi}_{m,j}\mathbf{C}_{m,k}\bigr)$. Under pilot contamination, additional off-diagonal contributions may also persist through nonzero cross-covariances between $\hat{\mathbf h}_{m,j}$ and $\hat{\mathbf h}_{m,j'}$. Since these residual cross-UE terms do not admit the same simplified structure as the dominant self-term, the CSI proxy in \eqref{eq:e_csi_proxy} retains only the power-independent self-term in \eqref{eq:app_csi_var_final}, while the remaining cross-UE CSI-leakage terms are omitted as a modeling approximation for analytical tractability. \hfill $\blacksquare$
	
	\vspace{-2mm}
	\section{Proof of Proposition~\ref{prop:convergence}}
	\label{app:convergence_proof}
	
	Under Assumption~\ref{ass:smoothness}, the $L_s$-smoothness descent lemma applied to the global update $\mathbf{w}^{(t+1)}=\mathbf{w}^{(t)}-\eta_g\hat{\mathbf{g}}^{(t)}$ gives
	\begin{equation}
		F(\mathbf{w}^{(t+1)}) \le F(\mathbf{w}^{(t)}) - \eta_g\bigl\langle \nabla F(\mathbf{w}^{(t)}),\, \hat{\mathbf{g}}^{(t)}\bigr\rangle + \tfrac{L_s\eta_g^2}{2}\bigl\|\hat{\mathbf{g}}^{(t)}\bigr\|^2.
		\label{eq:app_descent}
	\end{equation}
	Under Assumption~\ref{ass:clipping}, clipping bias is negligible, so the aggregated gradient estimate is conditionally unbiased: $\mathbb{E}[\hat{\mathbf{g}}^{(t)}\mid\mathbf{w}^{(t)}] = \nabla F(\mathbf{w}^{(t)})$. Applying the bias--variance decomposition $\mathbb{E}[\|\hat{\mathbf{g}}^{(t)}\|^2] = \|\nabla F(\mathbf{w}^{(t)})\|^2 + \mathrm{Var}[\hat{\mathbf{g}}^{(t)}]$ and taking expectations yields
	\begin{align}
		& \mathbb{E}\!\left[F(\mathbf{w}^{(t+1)})\right] \le \mathbb{E}\!\left[F(\mathbf{w}^{(t)})\right] \nonumber\\
		&- \eta_g\!\left(1-\tfrac{L_s\eta_g}{2}\right) \mathbb{E}\!\left[\|\nabla F(\mathbf{w}^{(t)})\|^2\right] + \tfrac{L_s\eta_g^2}{2}\,V,
		\label{eq:app_per_round}
	\end{align}
	where $V = \sigma_g^2/N + \mathrm{Var}_{\mathrm{agg}}$ collects the stochastic gradient variance (Assumption~\ref{ass:variance}) and the OTA aggregation variance from~\eqref{eq:ota_decomp}. The three independent noise sources, namely DP perturbation, multi-UE interference plus CSI error, and thermal noise contribute additively to $\mathrm{Var}_{\mathrm{agg}}$.
	
	For $\eta_g \le 1/L_s$, the coefficient satisfies $1 - L_s\eta_g/2 \ge 1/2$. Observing~\eqref{eq:app_per_round} over $t=0,\dots,T{-}1$ provides
	\begin{equation}
		\frac{\eta_g}{2}\sum_{t=0}^{T-1} \mathbb{E}\!\left[\|\nabla F(\mathbf{w}^{(t)})\|^2\right] \le \Delta_0 + \frac{L_s\eta_g^2 T}{2}\,V.
		\label{eq:app_telescope}
	\end{equation}
	Dividing both sides by $T\eta_g/2$ completes the proof. \hfill $\blacksquare$
	
	%%%%%%%%%% References %%%%%%%%%%%%%%
	\vspace{-3mm}
	\balance
	\bibliographystyle{IEEEtran}
	\bibliography{OTA_DP}

@IEEEtranBSTCTL{IEEEexample:BSTcontrol,
	CTLuse_url               = "no",
	CTLdash_repeated_names   = "no",
	CTLuse_forced_etal       = "yes",
	CTLmax_names_forced_etal = "1",
	CTLnames_show_etal       = "1"
}

@techreport{3gpp38901,
    author = {{3GPP}},
    title = {Study on channel model for frequency spectrum above 6 GHz},
    institution = {3rd Generation Partnership Project},
    type = {Technical Report},
    number = {TR 38.901},
    version = {16.1.0},
    date = {2020-03},
    year = {2020}
}

@article{lecun1998mnist,
  author  = {LeCun, Yann and Bottou, L{\'e}on and Bengio, Yoshua and Haffner, Patrick},
  title   = {Gradient-Based Learning Applied to Document Recognition},
  journal = {Proceedings of the IEEE},
  volume  = {86},
  number  = {11},
  pages   = {2278--2324},
  year    = {1998},
  doi     = {10.1109/5.726791}
}

@techreport{krizhevsky2009cifar,
  author      = {Krizhevsky, Alex and Hinton, Geoffrey},
  title       = {Learning Multiple Layers of Features from Tiny Images},
  institution = {University of Toronto},
  year        = {2009},
  type        = {Technical Report}
}

@article{sifaou2024ota,
  title={Over-the-air federated learning over scalable cell-free massive {MIMO}},
  author={Sifaou, Houssem and Li, Geoffrey Ye},
  journal={IEEE Transactions on Wireless Communications},
  volume={23},
  number={5},
  pages={4214--4228},
  year={2024},
  month={May},
  doi={10.1109/TWC.2023.3315962},
  publisher={IEEE}
}

@article{tavangaran2024privacy,
  title={On differential privacy for federated learning in wireless systems with multiple base stations},
  author={Tavangaran, Nima and Chen, Mingzhe and Yang, Zhaohui and Silva Jr., Jos Mairton B. and Poor, H. Vincent},
  journal={IET Communications},
  volume={18},
  number={20},
  pages={1853--1867},
  year={2024},
  month={January},
  doi={10.1049/cmu2.12722},
  publisher={The Institution of Engineering and Technology}
}

@article{mao2024providing,
  title={Providing differential privacy for federated learning over wireless: A cross-layer framework},
  author={Mao, Jiayu and Yin, Tongxin and Yener, Aylin and Liu, Mingyan},
  journal={arXiv preprint arXiv:2412.04408},
  year={2024},
  month={December}
}

@article{mcmahan2017communication,
  title={Communication-efficient learning of deep networks from decentralized data},
  author={McMahan, Brendan and Moore, Erika and Ramage, Daniel and Hampson, Seth and y Arcas, Blaise Aguera},
  journal={International Conference on Machine Learning (ICML)},
  pages={1273--1282},
  year={2017},
  organization={PMLR}
}

@article{dwork2014algorithmic,
  author    = {Cynthia Dwork and Aaron Roth},
  title     = {The Algorithmic Foundations of Differential Privacy},
  journal   = {Foundations and Trends in Theoretical Computer Science},
  volume    = {9},
  number    = {3-4},
  pages     = {211--407},
  year      = {2014},
  doi       = {10.1561/0400000042},
  publisher = {Now Publishers}
}

@inproceedings{bun2016concentrated,
  title={Concentrated differential privacy: Simplifications, extensions, and lower bounds},
  author={Bun, Mark and Steinke, Thomas},
  booktitle={Theory of Cryptography Conference (TCC)},
  pages={635--658},
  year={2016},
  organization={Springer},
  publisher={Springer, Berlin, Heidelberg}
}

@ARTICLE{batool2026privacy,
  author={Batool, Iqra and Fouda, Mostafa M. and Ibrahem, Mohamed I. and Fadlullah, Zubair Md},
  journal={IEEE Internet of Things Journal}, 
  title={Privacy-Preserving Federated Meta-Learning for Cell-Free Massive {MIMO}: Instant Adaptation with Distributed Intelligence}, 
  year={2026},
  volume={},
  number={},
  pages={1-1},
  doi={10.1109/JIOT.2026.3662248}}

@ARTICLE{yan2024over,
  author={Yan, Na and Wang, Kezhi and Pan, Cunhua and Chai, Kok Keong and Shu, Feng and Wang, Jiangzhou},
  journal={IEEE Transactions on Communications}, 
  title={Over-the-Air Federated Averaging With Limited Power and Privacy Budgets}, 
  year={2024},
  volume={72},
  number={4},
  pages={1998-2013},
  doi={10.1109/TCOMM.2023.3338738}}

@article{yang2020federated,
  title={Federated learning via over-the-air computation},
  author={Yang, Kai and Jiang, Tao and Shi, Yuanming and Ding, Zhi},
  journal={IEEE Transactions on Wireless Communications},
  volume={19},
  number={3},
  pages={2022--2035},
  year={2020},
  publisher={IEEE}
}

@article{amiri2020federated,
  title={Federated learning over wireless fading channels},
  author={Amiri, Mohammad Mohammadi and G{\"u}nd{\"u}z, Deniz},
  journal={IEEE Transactions on Wireless Communications},
  volume={19},
  number={5},
  pages={3546--3557},
  year={2020},
  publisher={IEEE}
}

@inproceedings{abadi2016deep,
  title={Deep learning with differential privacy},
  author={Abadi, Martin and Chu, Andy and Goodfellow, Ian and McMahan, H Brendan and Mironov, Ilya and Talwar, Kunal and Zhang, Li},
  booktitle={Proceedings of the 2016 ACM SIGSAC Conference on Computer and Communications Security},
  pages={308--318},
  year={2016}
}

@article{liu2021privacy,
  title={Privacy for free: Wireless federated learning via uncoded transmission with adaptive power control},
  author={Liu, Dongzhu and Simeone, Osvaldo},
  journal={IEEE Journal on Selected Areas in Communications},
  volume={39},
  number={1},
  pages={170--185},
  year={2021},
  publisher={IEEE}
}

@article{ngo2017cell,
  title={Cell-free massive {MIMO} versus small cells},
  author={Ngo, Hien Quoc and Ashikhmin, Alexei and Yang, Hongyang and Larsson, Erik G and Marzetta, Thomas L},
  journal={IEEE Transactions on Wireless Communications},
  volume={16},
  number={3},
  pages={1834--1850},
  year={2017},
  publisher={IEEE}
}

@article{buzzi2017cell,
  title={Cell-free massive {MIMO}: User-centric approach},
  author={Buzzi, Stefano and D'Andrea, Carmen},
  journal={IEEE Wireless Communications Letters},
  volume={6},
  number={6},
  pages={706--709},
  year={2017},
  publisher={IEEE}
}

@article{bjornson2019making,
  title={Making cell-free massive {MIMO} competitive with {MMSE} processing and centralized implementation},
  author={Bj{\"o}rnson, Emil and Sanguinetti, Luca},
  journal={IEEE Transactions on Wireless Communications},
  volume={19},
  number={1},
  pages={77--90},
  year={2019},
  publisher={IEEE}
}

@article{wei2022low,
  title={Low-latency federated learning over wireless channels with differential privacy},
  author={Wei, Kang and Li, Jun and Ding, Ming and Ma, Chuan and Su, Hang and Zhang, Bingsheng and Poor, H Vincent},
  journal={IEEE Journal on Selected Areas in Communications},
  volume={40},
  number={1},
  pages={290--307},
  year={2022},
  publisher={IEEE}
}

@article{jiang2023optimized,
  title={Optimized power control for over-the-air federated averaging with data privacy guarantee},
  author={Jiang, Jiahua and Han, Kailin and Du, Yulin absorption and Zhu, Guangxu and Wang, Zhiyong and Cui, Shuguang},
  journal={IEEE Transactions on Vehicular Technology},
  volume={72},
  number={2},
  pages={2728--2733},
  year={2023},
  publisher={IEEE}
}

@article{wang2022interference,
  title={Interference management for over-the-air federated learning in multi-cell wireless networks},
  author={Wang, Zhaohui and Zhou, Yijie and Shi, Yuanming and Zhuang, Weihua},
  journal={IEEE Journal on Selected Areas in Communications},
  volume={40},
  number={8},
  pages={2361--2377},
  year={2022},
  publisher={IEEE}
}

@article{gesbert2010multi,
  title={Multi-cell {MIMO} cooperative networks: A new look at interference},
  author={Gesbert, David and Hanly, Stephen and Huang, Howard and Shamai Shitz, Shlomo and Simeone, Osvaldo and Yu, Wei},
  journal={IEEE Journal on Selected Areas in Communications},
  volume={28},
  number={9},
  pages={1380--1408},
  year={2010},
  publisher={IEEE}
}

@article{cao2022optimized,
  title={Optimized power control design for over-the-air federated edge learning},
  author={Cao, Xiaota and Zhu, Guangxu and Xu, Jie and Wang, Zhiyong and Cui, Shuguang},
  journal={IEEE Journal on Selected Areas in Communications},
  volume={40},
  number={1},
  pages={342--358},
  year={2022},
  publisher={IEEE}
}
	
\end{document}